\documentclass[12pt]{iopart}
\pdfoutput=1
\usepackage{hyperref}

\usepackage{iopams}
\usepackage{amsopn}
\usepackage{amsthm}
\usepackage{amssymb}
\usepackage{dsfont}

\usepackage{graphicx}
\usepackage{caption}
\usepackage{subcaption}
\usepackage{epstopdf}

\usepackage{units}

\usepackage{array}
\usepackage{verbatim,indentfirst}
\usepackage[title,titletoc]{appendix}
\usepackage[nottoc,notlot,notlof]{tocbibind}

\usepackage[square,sort&compress,numbers]{natbib}

\newcounter{tlc}

\newtheorem{theorem}[tlc]{Theorem}
\newtheorem{lemma}[tlc]{Lemma}
\newtheorem{corollary}[tlc]{Corollary}

\newtheorem{definition}{Definition}
\newtheorem*{theorem*}{Theorem}
\newtheorem*{lemma*}{Lemma}
\newtheorem*{corollary*}{Corollary}

\makeatletter
\renewcommand{\p@subsection}{}
\renewcommand{\p@subsubsection}{}
\makeatother

\newcommand*{\ket}[1]{| #1 \rangle}

\newcommand*{\ketbra}[2]{| #1 \rangle\!\langle #2 |}

\newcommand{\kB}{k_\mathrm{B}}
\newcommand{\fwd}{\mathrm{fwd}}
\newcommand{\rev}{\mathrm{rev}}

\newcommand{\id}{\mathds{1}}

\DeclareMathOperator\mod{mod}

\usepackage{stmaryrd}
\makeatletter
\providecommand\underarrow@[3]{%
  \vtop{\vspace*{-0.95em}\ialign{##\crcr$\m@th\hfil#2{\;}{\scriptstyle #3}{\;}\hfil$\crcr
  \noalign{\nointerlineskip\kern.10\baselineskip}#1#2\crcr}}}
\providecommand{\underrightarrow}{%
  \mathpalette{\underarrow@\rightarrowfill@}}
\providecommand\rightarrowfill@{\arrowfill@\Mapstochar\relbar\rightarrow}
\providecommand\arrowfill@[4]{%
  $\m@th\thickmuskip0mu\medmuskip\thickmuskip\thinmuskip\thickmuskip
   \relax#4#1\mkern0mu%
   \cleaders\hbox{$#4\mkern-4mu#2\mkern-2mu$}\hfill
   \mkern-5mu#3$%
}

\makeatother

\newcommand{\caphead}[1]{{\bf #1}}
\newcommand{\figco}{} 

\newcommand{\strong}[1]{\textbf{#1}}

\usepackage{chngcntr}
\let\oldappendices\appendices
\def\appendices{\def\theequation{\Alph{section}\arabic{equation}}%
\counterwithin*{equation}{section}
\oldappendices}

\begin{document}

%
%
\newcommand{\ourtitle}{Introducing one-shot work into fluctuation relations}
\title[\ourtitle]{\ourtitle}

%
%
\author{Nicole~Yunger~Halpern$^1$, Andrew~J.~P.~Garner$^{23}$, Oscar~C.~O.~Dahlsten$^2$, and Vlatko~Vedral$^{23}$}
\address{$^1$ Institute for Quantum Information and Matter, Caltech, \\ Pasadena, CA 91125, USA}
\address{$^2$ Atomic and Laser Physics, Clarendon Laboratory,
University of Oxford, Parks~Road, Oxford, OX13PU, United Kingdom}
\address{$^3$ Center for Quantum Technologies, National University of Singapore,
3~Science~Drive~2, 117543, Republic~of~Singapore}

\eads{nicoleyh@caltech.edu}
\date{\today}

%
%
%
%
\begin{abstract}
Two approaches to small-scale and quantum thermodynamics are fluctuation relations and one-shot statistical mechanics. Fluctuation relations (such as Crooks' Theorem and Jarzynski's Equality) relate nonequilibrium behaviors to equilibrium quantities such as free energy. 
One-shot statistical mechanics involves statements about every run of an experiment, not just about averages over trials.

We investigate the relation between the two approaches. 
We show that both approaches feature the same notions of work and the same notions of probability distributions over possible work values. 
The two approaches are alternative toolkits with which to analyze these distributions.
To combine the toolkits, we show how one-shot work quantities can be defined and bounded in contexts governed by Crooks' Theorem.
These bounds provide a new bridge from one-shot theory to experiments originally designed for testing fluctuation theorems.
\end{abstract}




\maketitle

\tableofcontents

\clearpage
\newpage

\markboth{\ourtitle}{\ourtitle}

\section{Introduction}
The probabilistic nature of thermalization prevents us from deterministically predicting the amount of work performed on a system in any given run of an experiment. 
This stochasticity necessitates a statistical treatment of work, especially when the deviation from the mean value of work is large.
Two popular frameworks employed for this purpose are {\em fluctuation theorems}~\cite{Jarzynski97,Crooks99,Kurchan00, Tasaki00, EngelN07, TalknerLH07, TalknerH07, QuanD08, CampisiTH09, TalknerCH09, CampisiHT11, EspositoHM09, HideV10, CohenI12, DornerCHFGV13} and {\em one-shot statistical mechanics}~\cite{Renner05,DahlstenRRV11,EgloffDRV12,Aberg13,HorodeckiO13,SkrzypczykSP14,BrandaoHNOW15,YungerHalpernR14}.
The former framework's purpose is to quantify the behaviors of nonequilibrium classical and quantum systems. 
The latter framework concerns statements true of every trial (realization) of an experiment.

The relationship between these frameworks has been unclear (though work by \AA{}berg~\cite{Aberg13} suggests that a connection could be fruitful).
We will demonstrate that these approaches are not competitors. Rather, the approaches are mutually compatible tools. 
Combined, they describe general thermal behaviors of small classical and quantum systems.

\begin{figure}[hbt]
\centering
\includegraphics[width=.65\textwidth, clip=true]{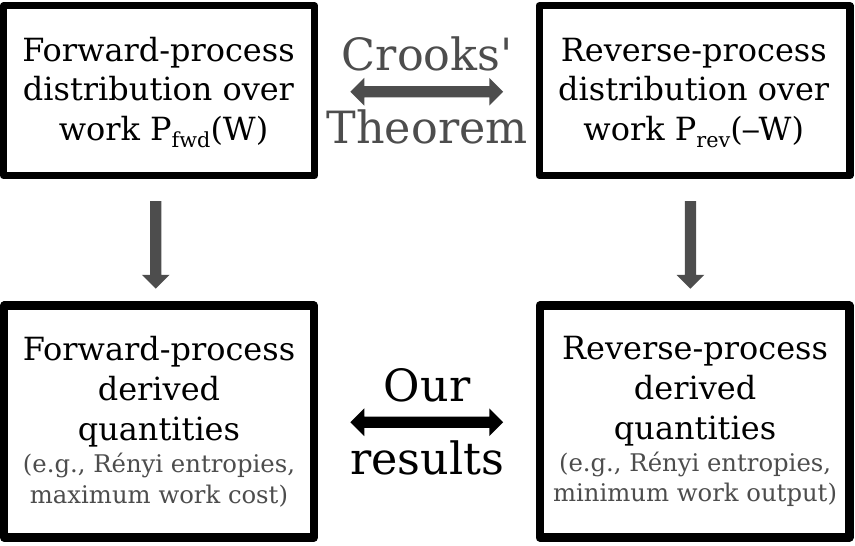}
\caption{
\caphead{A synopsis of how our results relate to Crooks' Theorem.}
Crooks' fluctuation theorem links a probability distribution $P_\fwd(W)$ over work expended during one process to the distribution $P_\rev(-W)$ over the work recouped during the reverse process. 
One-shot statistical mechanics concerns functions of probability distributions, such as R\'{e}nyi entropies.
We demonstrate how one-shot tools can be applied to problems governed by Crooks' Theorem.
This application allows us to calculate one-shot properties of the forward protocol from properties of the reverse, without the need to profile entire probability distributions. 
}
\label{fig:Synopsis}
\end{figure}

We will begin with a technical introduction to fluctuation theorems and one-shot statistical mechanics.
We then present our main claim: {\em that one-shot statistical mechanics can be applied to settings governed by fluctuation theorems.}
We substantiate this claim by generalizing the characteristic functions of the work probability distributions for classical and quantum systems. 
From this generalization, we derive bounds on one-shot work quantities in settings governed by fluctuation theorems.
We demonstrate how this generalization can be employed in two mathematical formalisms: a {\em work-extraction game}~\cite{EgloffDRV12,Aberg13} and {\em thermodynamic resource theories}~\cite{JanzingWZGB00,HorodeckiHHH09,BrandaoHORS13,HorodeckiO13,YungerHalpernR14,BrandaoHNOW15}.
To conclude with two pedagogical examples, we apply the generalization to specific fluctuation settings: {\em Landauer bit reset}~\cite{Szilard29,Landauer61,JanzingWZGB00,EgloffDRV12,Aberg13,BrowneGDV14} and {\em experimental DNA unzipping}~\cite{MossaMFHR09,ManosasMFHR09,AlemanyR10}.
The examples illustrate the opportunity to test one-shot results with experiments devised originally for fluctuation theorems.

\section{Preliminaries}
\subsection{Fluctuation theorems}
Consider a classical system that is coupled to a heat bath and driven externally.
Due to the probabilistic nature of thermalization, the amount of heat transferred between the system and the bath in any given trial cannot be predicted.
Hence the amount of work done by the drive, in any given trial, cannot be predicted.
The protocol can be associated with a {\em work distribution} $P(W)$, the probability density associated with some trial's costing an amount $W$ of work.
In equilibrium thermodynamics, $P(W)$ peaks tightly at the average value 
$W = \langle W \rangle$.
This value suffices to describe the work performed in each trial.
In more general, nonequilibrium, thermodynamics, the average does not suffice.
Yet thermodynamic state variables related to averages (temperature, free energy, etc.) are used in nonequilibrium thermodynamics. Fluctuation relations link these variables to probability distributions over work or heat.
We will focus mostly on continuous variables $W$, which have been used in classical and quantum contexts\footnote{
$W$ can be continuous even if the system has a discrete energy spectrum: Consider a system that interacts with the heat bath while two (or more) energy levels shift at different rates.
The system can jump from one level to another at any time, due to thermalization. The work cost of a trajectory that jumps at time $t$ can differ infinitesimally from the work cost of a trajectory that jumps at time $t+dt$.
} (e.g.,~\cite{TalknerLH07,QuanD08,TalknerCH09,DornerCHFGV13}).

One such fluctuation relation is Crooks' Theorem~\cite{Crooks99}.
Though originally derived in a classical setting, it has been shown to govern quantum processes~\cite{Kurchan00, Tasaki00, EngelN07, TalknerLH07, TalknerH07, QuanD08, CampisiTH09, TalknerCH09, CampisiHT11, EspositoHM09, CohenI12, DornerCHFGV13}. 
Crooks' Theorem describes the fluctuations in the work expended on systems subject to a time-changing Hamiltonian $H(\lambda_t)$ in the presence of a heat bath. An experimenter can change the external scalar parameter $\lambda_t$ during the time interval $t \in [-\tau, \tau]$ by performing work.
We denote the bath's inverse temperature by  $\beta = 1 / \kB T $.
The external driving can be performed in either a forward or a reverse direction. 
The forward process begins at time $-\tau$, when the system occupies the thermal state $e^{-\beta H_{-\tau} } / Z_{-\tau}$  of the initial Hamiltonian 
$H_{-\tau} \equiv H( \lambda_{-\tau} )$. 
The reverse process begins at time $\tau$, when the system occupies the thermal state $e^{ - \beta H_\tau } / Z_\tau$ associated with the final Hamiltonian $H_\tau  \equiv  H( \lambda_\tau )$ of the forward protocol.
($Z_{ \pm \tau }$ denote normalization factors.)

Suppose that an agent implements the protocol in both directions many times, measuring the work invested in each forward trial and the work extracted from each reverse.
Two probability distributions encapsulate these measurements: $P_\fwd (W)$ denotes the probability that some forward trial will require work $W$ (or the probability per unit work, if $P_\fwd$ denotes a probability density), and $P_\rev(-W)$ denotes the probability that some reverse trial will output work $W$. 

If the system's interactions with the bath are Markovian and {\em microscopically reversible},~\cite{Crooks98} and if the initial state is thermal, the work probability distributions satisfy \emph{Crooks' Theorem}~\cite{Crooks99},
\begin{equation}   \label{eq:CrooksThm}
   \frac{   P_\fwd(W)   }{   P_\rev(-W)   }
   =   e^{ \beta ( W - \Delta F ) }.
\end{equation}
This $\Delta F$ denotes the difference 
$F ( e^{-\beta H_{\tau} } / Z_{\tau} )  -  F ( e^{-\beta H_{-\tau} } / Z_{-\tau} )$
between the Helmholtz free energies between thermal states over the forward process's final and initial Hamiltonians.

Multiplying each side of Crooks' Theorem by $P_\rev(-W) e^{- \beta W}$ and integrating over $W$ yields \emph{Jarzynski's Equality}~\cite{Jarzynski97},
\begin{equation}   \label{eq:JarzEq}
   \left\langle   e^{ - \beta W  }  \right\rangle_{\rm fwd}
   =   e^{ - \beta \Delta F },
\end{equation}
wherein $\langle . \rangle_{\rm fwd}$ denotes an expectation value calculated from $P_{\rm fwd}$. 
Applied to a work distribution P(W) constructed from simulations or experiments, Jarzynski's Equality can be used to calculate the equilibrium quantity $\Delta F$. 
Combined with Jensen's Inequality, $\langle e^{x} \rangle   \geq   e^{ \langle x \rangle }$, Jarzynski's Equality implies a lower bound on the average work required to complete a trial. 
This bound, $\langle W \rangle \geq \Delta F$, has been considered a statement of the Second Law of Thermodynamics~\cite{Jarzynski08}.

The left-hand side (LHS) of Jarzynski's Equality has been recognized as the characteristic function, or Fourier transform, of $P_{\rm fwd}(W)$~\cite{TalknerLH07,DornerCHFGV13}. 
If $u = i \beta$ denotes the variable conjugate to $W$, the characteristic function is
\begin{equation} \label{eq:ChiFwd}
\hspace{-2em}    \chi_{\rm fwd}(\beta)
    \equiv   \mathcal{F} \{ P_{\rm fwd}(W) \}
    \equiv
    \int_{ -\infty }^\infty   dW   P_\fwd (W)    e^{ i u W }
    =   \int_{ -\infty }^\infty   dW
           P_\fwd (W)    e^{ - \beta W }.
\end{equation}
In terms of the characteristic function, Jarzynski's Equality reads,
\begin{equation}   \label{eq:JarzGFwd}
   \chi_{\rm fwd}(\beta)   =   e^{-\beta \Delta F}.
\end{equation}
The reverse process corresponds to the characteristic function 
$ \chi_{\rm rev}(\beta)
   \equiv   \int_{-\infty}^\infty   dW   P_{\rm rev}(W)   e^{- \beta W}$, in terms of which 
   $\chi_{\rm rev}( \beta )   =   e^{ \beta \Delta F}$.

\subsection{One-shot statistical mechanics}
Mean values do not necessarily reflect a system's {\em typical} behavior.
Consider a system that must output at least some threshold amount of work to trigger another process. One such threshold is the activation energy required to begin a chemical reaction. 
The system might output below-threshold work usually but far-above-threshold work occasionally. 
The average work might exceed the threshold, but the second process is usually not trigged.

By spotlighting statistics other than the mean, one-shot information theory extends idealized protocols implemented $n \to \infty$ times to realistic finite-$n$ protocols that might fail. 
Conventional statistical mechanics describes the optimal rate at which work can be extracted \emph{asymptotically}. 
Consider transforming $n$ copies of one equilibrium state into $n$ copies of another quasistatically, in the presence of a temperature-$T$ heat bath. In the asymptotic, or thermodynamic, limit as $n \to \infty$, the average work required per copy approaches the difference $\Delta F$ between the states' free energies.
The free energy depends on the Shannon entropy.
In reality, states are transformed finitely many times, and realistic processes have probabilities $\delta$ of failing to accomplish their purposes. 
Finite-$n$ work-consumption rates have been quantified with \emph{one-shot entropies}~\cite{DahlstenRRV11,HorodeckiO13,SkrzypczykSP14,BrandaoHNOW15,YungerHalpernR14}.
So have the efficiencies of finite-$n$ data compression, randomness extraction, quantum key distribution, and hypothesis testing~\cite{Renner05,RennerW04,Tomamichel12,DupuisKFRR12}.

One one-shot entropy is the order-$\infty$ R\'enyi entropy $H_\infty ( \mathcal{P} ) $, known also as the {\em min-entropy}. For any discrete probability distribution $\mathcal{P}$ whose greatest element is $\mathcal{P}^{\rm max}$,
\begin{equation}   \label{eq:HMinDiscrete}
   H_\infty (\mathcal{P}) \equiv   - \log ( \mathcal{P}^{\rm max} ).
\end{equation}
(All logarithms in this article are base-$e$.)
We will discuss two popular models in which one-shot entropies are applied to thermodynamics: {\em work-extraction games} and {\em thermodynamic resource theories}.

\subsubsection{Work-extraction game}
\label{sec:WEGameIntro}
~\\
In the work-extraction game described by Egloff \emph{et al.}~\cite{EgloffDRV12}, a player transforms a system in a state $\rho$, governed by a Hamiltonian $H_\rho$, into a state $\sigma$ governed by $H_\sigma$:
$(\rho, H_\rho)   \mapsto   (\sigma, H_\sigma)$.
For simplicity, we take a semiclassical model such that states are assumed to commute with their Hamiltonians.
The agent has access to a temperature-$T$ heat bath.

The player should choose an {\em optimal strategy} to maximize the  transformation's work output (or minimize the transformation's work cost).
The strategy consists of a sequence of operations of two types: (1) Without investing work, the player can couple the system to the bath in any manner modeled by a stochastic matrix that preserves the Gibbs state $e^{ - \beta H_\rho } / Z_\rho$. 
(Such thermalization models are discussed in Appendix~\ref{section:ThermModelsMain}.)
(2) By investing or extracting work, the agent can shift the Hamiltonian's levels.

The primary result in~\cite{EgloffDRV12} implies an upper bound on the work extractable (up to a probability $\delta$ of failure) during the transformation $(\rho, H_\rho)   \mapsto   (\sigma, H_\sigma)$. 
Egloff \emph{et al.}\ show that the {\em optimal strategy} has a probability $1 - \delta$ of outputting at least the work
\begin{equation}   \label{eq:EgloffDRV12Thm}
   {w}^{\delta}_{\rm best} ( \rho, H_\rho   \mapsto   \sigma, H_\sigma )
   =  \kB T  \log  \left(   M   \left(   
       \frac{  G^T(\rho) }{  1 - \delta  }  \:  ||  \:  G^T( \sigma )  
       \right)   \right)
\end{equation}
in each trial. 
$G^T$ denotes \emph{Gibbs-rescaling} relative to the temperature $T$. Gibbs-rescaling facilitates the comparison of the work values of states governed by different Hamiltonians. 
A state can have the capacity to perform work due to the state's information content (e.g., because the state is pure) and energy contents (e.g., because the state has weight on high energy levels). 
Gibbs-rescaling recasts the state's work capacity as entirely informational. 
This recasting facilitates the comparison of states governed by different Hamiltonians. $M$ denotes the \emph{relative mixedness}, a measure of how much more mixed one state is, or how much less information-sourced work capacity a state has. 
Dissipative processes yield less than the optimal amount $w^\delta_{\rm best}$ of work. 
Hence Eq.~(\ref{eq:EgloffDRV12Thm}) upper-bounds the amount that could be extracted with an arbitrary (possibly suboptimal) strategy.

\subsubsection{Thermodynamic resource theories}
~\\
Resource theories have been used to calculate how efficiently scarce quantities can be distilled and converted into other forms via cheap (or ``free'') operations~\cite{HorodeckiHHH09}. To an agent able to perform only certain operations for free, each state has some value, or worth. We can quantify this value with resource theories.
\emph{Thermodynamic resource theories} model exchanges of heat amongst systems and  baths~\cite{BrandaoHORS13,HorodeckiO13,BrandaoHNOW15,JanzingWZGB00,YungerHalpernR14}. 
Each resource theory is defined by the inverse temperature $\beta$ of a heat bath from which the agent can draw Gibbs states for free. 
More generally, energy-conserving \emph{thermal operations} can be performed for free. Nonequilibrium states have value because work can be extracted from them.

Horodecki and Oppenheim introduced one-shot tools into thermodynamic resource theories~\cite{HorodeckiO13}. 
They focused on semiclassical resource theories, in which states commute with the Hamiltonians that govern them.
Horodecki and Oppenheim calculated the minimum work required to create a state within trace distance $\varepsilon$ of a target state $\rho$. 
They also analyzed the transfer of work from $\rho$ to a battery defined by a Hamiltonian of gap $w$. 
The maximum $w$ such that the battery ends within trace distance $\delta$ of its excited state was shown to be related to a one-shot entropy of $\rho$. 
One-shot information theory has since been applied to \emph{catalysis} (the facilitation of a transformation by an ancilla)~\cite{BrandaoHORS13,BrandaoHNOW15}, to arbitrary baths such as particle baths~\cite{YungerHalpernR14}, and to quantum problems (that involve states that do not commute with the Hamiltonian)~\cite{LostaglioJR15,LostaglioKJR15}.

\section{Unification of fluctuation theorems and one-shot statistical mechanics}
\label{sec:Unification}
Fluctuation theorems and one-shot statistical mechanics concern properties of work distributions beyond averages.
The two frameworks do not compete to describe the same concept in alternative ways. 
Rather, the frameworks complement each other and can be combined into a general description of small-scale classical and quantum systems.
Fluctuation theorems are restricted to systems that satisfy certain physical assumptions and that can undergo forward and reverse protocols. Crooks' Theorem~\cite{Crooks99}, for example, relies on the dynamics' Markovianity and microscopic reversibility, and on the system's beginning in a thermal state. 
The tools of one-shot statistical mechanics (e.g.\ R\'enyi entropies, and bounds on work values in every trial of an experiment) can be applied more generally to the statistics produced by any system that consumes work.
The formalisms are not incompatible: {\em The tools of one-shot statistical mechanics can be applied to the work distribution of any process governed by Crooks' Theorem}.

We will substantiate this claim by focusing on the one-shot concept of {\em guaranteed work}: an upper bound (up to some error) on the work required to complete some process that applies not just on average, but in every trial.
We will define this quantity in contexts governed by Crooks' Theorem and will relate the quantity to the one-shot entropy $H_\infty$. 
Our results describe all quantum and classical systems whose thermalization satisfies microscopic reversibility and Markovianity and whose work distribution is continuous. 
(For details about these assumptions' realizations in two common one-shot frameworks, see Appendix~\ref{section:ThermModelsMain}.)

\subsection{One-shot work quantities in fluctuation contexts}
\label{section:WorkDefns}

\begin{figure}[tb]
\centering
\includegraphics[width=.7\textwidth, clip=true]{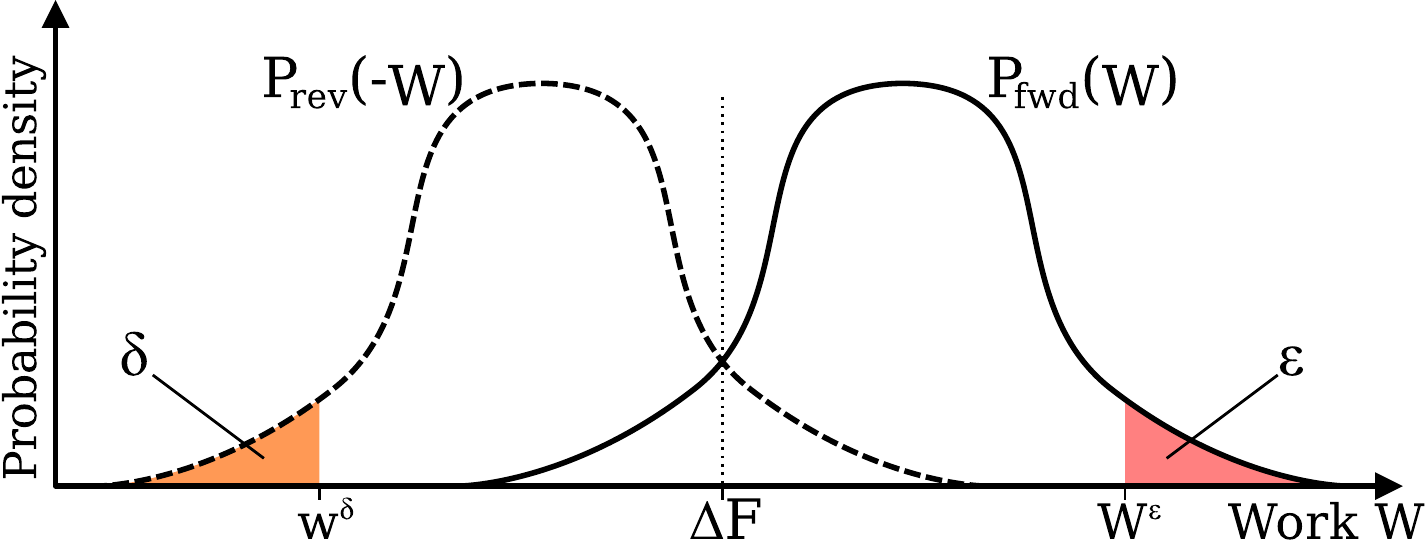}
\caption{
\figco
\caphead{One-shot work quantities $W^\varepsilon$ and $w^{\delta}$.}
The forward and reverse work distributions $P_\fwd(W)$ and $P_\rev(-W)$ intersect at $\Delta F$, the difference between the free energies of the Gibbs states associated with the forward protocol's initial and final Hamiltonians.
The shaded region under the right tail of the $P_\fwd(W)$ curve has an area $\varepsilon$, which is the probability that a forward process consumes more work than the {\em $\varepsilon$-required work} $W^\varepsilon$.
The shaded region under the left tail of the $P_\rev(-W)$ curve has an area $\delta$, which is the probability that the reverse process outputs less work than the {\em $\delta$-extractable work} $w^{\delta}$.
}
\label{fig:Distributions}
\end{figure}

Suppose we consider the behavior of a system evolving under a process that is driven by a single external parameter, and otherwise satisfies the conditions for Crooks' theorem to hold. 
For clarity in this article, we shall choose examples where the {\em forward process} tends to cost work to complete.
We will upper-bound the amount of work required to complete a single trial of a process successfully, such that this bound is only exceeded with probability $\varepsilon$ (as illustrated on the right-hand side (RHS) of Fig.~\ref{fig:Distributions}).

\begin{definition} 
\label{def:WEps}
Each implementation of the forward protocol has a probability $1 - \varepsilon$ of requiring no more work than the \mbox{\strong{$\mathbf{\varepsilon}$-required work}} $W^\varepsilon$ that satisfies
\begin{equation}   \label{eq:WEpsDefn1}
   \int_{-\infty}^{ W^\varepsilon }   dW   P_\fwd(W)  
   =   1 - \varepsilon.
\end{equation}
The trial has a probability ${\varepsilon} \in [0, 1]$ of requiring more work than $W^\varepsilon$.
\end{definition}

Similarly, we will lower-bound the amount of work extracted in the reverse process, for all but $\delta$ of the trials (illustrated on the LHS of Fig.~\ref{fig:Distributions}).

\begin{definition}
\label{def:wDelta}
Each implementation of the reverse protocol has a probability $1 - \delta$ of outputting at least the \mbox{\strong{${\delta}$-extractable work}} $w^\delta$ that satisfies
\begin{equation}   \label{eq:wEpsDefn1}
   \int_{ w^\delta }^\infty   dW   P_\rev(-W)  
   =   1  -  {\delta}.
\end{equation}
The trial has a probability $\delta  \in [0, 1]$ of outputting less work than $w^\delta$.
\end{definition}

The failure probability has two interpretations.
Suppose, in the work-investment case, that an agent invests only the amount $W^\varepsilon$ of work in a forward trial. 
The external parameter $\lambda_t$ has a probability $\varepsilon$ of failing to reach $\lambda_\tau$. 
Alternatively, suppose the agent invests all the work required to evolve $\lambda_t$ to $\lambda_\tau$. 
The agent has a probability $\varepsilon$ of overshooting the ``work budget'' $W^\varepsilon$.  
The failure probability $\delta$ associated with work extraction can be interpreted similarly.

\subsection{One-shot Jarzynski equalities} 

\label{sec:OneShotCFT}
Even if only forward trials have been performed, the reverse process's $w^\delta$ can be calculated from Crooks' Theorem. 

\begin{lemma}   \label{lemma:wEpsCalcn}
Each reverse trial has a probability $1 - {\delta}$ of outputting at least the amount $w^{\delta}$ of work that satisfies 
\begin{equation}   \label{eq:wEpsCalcn}
   \chi^{ \delta }_\fwd ( \beta )
    = (1   -   \delta)  e^{ - \beta \Delta F },
\end{equation}
wherein
\begin{equation}  \label{eq:ChiDDef}
   \chi^{ \delta }_\fwd (\beta)
   \equiv
   \int_{ w^\delta }^\infty   dW     P_\fwd (W)   e^{ - \beta W}
\end{equation}
generalizes the characteristic function $\chi_{\rm fwd}(\beta)$.

\begin{proof}
Upon multiplying each side of Crooks' Theorem [Eq.~(\ref{eq:CrooksThm})]  by 
$P_\rev(-W) e^{ - \beta W }$, we integrate from $w^\delta$ to infinity. 
The LHS equals $\chi^\delta_\fwd(\beta)$ by definition [Eq.~(\ref{eq:ChiDDef})].
The right-hand integral evaluates to $1 - \delta$ by Definition~\ref{def:wDelta}.
\end{proof}
\end{lemma}

We can calculate $W^\varepsilon$ from $P_{\rm rev}(-W)$ via Crooks' Theorem in the same way:

\begin{lemma}   
\label{lemma:WEpsCalcnCrooks}
Each forward trial has a probability $1 - \varepsilon$ of requiring no more work than the $W^\varepsilon$ that satisfies
\begin{equation}   \label{eq:WEpsCalcnCrooks}
   \chi^{ \varepsilon }_\rev ( \beta )
   =   (1  -  \varepsilon)   e^{ \beta \Delta F },
\end{equation}
wherein
\begin{equation}   \label{eq:ChiEpsRev}
   \chi^{ \varepsilon }_\rev (\beta)
   \equiv
   \int_{ -W^\varepsilon }^\infty   dW      P_\rev (W)   e^{ - \beta W}
\end{equation}
generalizes the characteristic function $\chi_{\rm rev}(\beta)$.
\end{lemma}

These one-shot generalizations extend Jarzynski's Equality [Eq.~\eref{eq:ChiFwd}], rendering it more robust against unlikely (probability less than $\epsilon$) but highly expensive (work cost more than $W^\varepsilon$) fluctuations in work. 
The generalizations characterize every quantum or classical process that produces a work distribution governed by Crooks' theorem.
An alternative proof of these general lemmata---specialised for quantum systems undergoing unitary evolution and hence producing discrete work distributions---is presented in Appendix~\ref{app:quantJE}.

\subsection{Bounding one-shot work quantities}
We can use Crooks' Theorem, via the above lemmata, to derive bounds on the one-shot required and extractable work. 
These bounds depend on characteristics of the work distributions:

\begin{theorem}
\label{theorem:wDeltaBound}
The work $\delta$-extractable from each reverse trial satisfies
\begin{equation}   \label{eq:wEpsBound}
   w^{\delta}   \leq
   \Delta F   -   \frac{1}{\beta}   
          [   {H^\beta_\infty} ( P_\fwd )   +   \log ( 1 - {\delta} )   ]
\end{equation}
for $\delta \in [0, 1)$, wherein we have defined
\begin{equation}   \label{eq:CtsHDefn}
   {H^\beta_\infty} ( P )  
    \equiv   - \log ( P^{\rm max}  /  \beta )
\end{equation}
for continuous work distributions.

\begin{proof}
Let $P^{\rm max}_\fwd$ denote the greatest value of $P_\fwd(W)$:
$P^{\rm max}_\fwd   \geq   P_\fwd(W)   \; \: \forall \: W$.
We can upper-bound the integral implicit in the $\chi^\delta_{\rm fwd}( \beta )$ of Eq.~(\ref{eq:wEpsCalcn}) in Lemma~\ref{lemma:wEpsCalcn}:
\begin{eqnarray}
   (1 - \delta )   e^{ - \beta   \Delta F }
  &   =   \chi^\delta_{\rm fwd}( \beta ) 
        \equiv   \int_{ w^\delta }^\infty   dW
                     P_\fwd (W)   e^{ - \beta W}  \nonumber \\
  &   \leq    P_\fwd^{\rm max}   \int_{ w^\delta }^\infty   dW   
                  e^{ - \beta W}  \nonumber \\
  &    =   e^{ \log \left( { P_\fwd^{\rm max} }  / {\beta}   \right) } 
              e^{ - \beta w^\delta } \nonumber \\
  & =   e^{ - {H^\beta_\infty} ( P_{\rm fwd} )     - \beta w^\delta }.
\end{eqnarray}
Solving for $w^\delta$ yields Ineq. (\ref{eq:wEpsBound}).
\end{proof}
\end{theorem}

Analogous statements, which we present without further proof, describe $W^\varepsilon$:
\begin{theorem} 
\label{theorem:WepsBound}
The work $\varepsilon$-required during each forward trial is bounded by
\begin{equation} \label{eq:CrooksWLemma}
   W^\varepsilon  \geq
   \Delta F   +   \frac{1}{\beta}
   [ {H^\beta_\infty}( P_\rev )   +   \log( 1 - \varepsilon  ) ]
\end{equation}
for failure probability $\varepsilon \in [0, 1)$. 
\end{theorem}

These theorems demonstrate our central claim: that one-shot statistical mechanics can be applied to settings governed by fluctuation theorems.
We have related the one-shot work quantities $W^{\epsilon}$ and $w^{\delta}$ to the one-shot entropy ${H^\beta_\infty}$.

Operationally when one handles data arising from a simulation or experiment, one does not directly observe a work distribution.
Rather, one obtains a list of values that could be divided into bins of finite range (i.e.,\ presented as a histogram) to approximate the true probability distribution.
As want to consider a theoretical bound that reflects the behavior of the underlying thermal process independently of the choice of binning, we consider the entropy expressed in terms of the probability density $P(W)$.
For Theorems \ref{theorem:wDeltaBound} and \ref{theorem:WepsBound} to be applicable, it must be the case that the experimental data can be fitted to a theoretical model- a weak assumption for any scientific experiment.

There are further considerations that must be taken into account when considering the entropy of a continuous distribution in contrast to the entropy of a histogram taken from that distribution.
(Throughout the following paragraph, we refer to $P_\fwd$ and $P_\rev$ jointly as $P$.) 
For a histogram taken of $P(W)$,
in the limit of small enough bin width $dW$, the probability of each bin is well approximated by the product $P(W) dW$ as evaluated at a point $W$ inside that bin.
However, in the limit of decreasing bin size, the probability of being in any particular bin will tend to zero, and the order-$\infty$ entropy, as previously defined, diverges:
$H_\infty(P) 
= \lim_{dW\to0} 
\left[ -\log\left(P^{\rm max} dW\right) \right] \to \infty$.
An unmodified $H_\infty$ is not a useful quantity in this limit.
By this measure, there is an infinite amount of information in {\em any} continuous distribution, and hence it can not be used to quantify the amount by which some continuous distributions are more entropic than others.
To circumvent this problem, we employ a type of renormalization. $H_\infty$ can be split into a finite part that varies with the distribution under consideration, and an infinite part that does not:
$H_\infty = -\log\left(P^{\rm max} dW\right) = -\log\left(P^{\rm max} k^{-1}\right) - \log\left(k dW\right)$, where $k$ is an arbitrary factor with units of inverse energy chosen such that the argument of each logarithm is dimensionless.
When one considers the {\em difference} in entropy between distributions, the latter term always cancels out.
As such the quantity $H^k_\infty= -\log\left(P^{\rm max} k^{-1}\right) $, by omitting the latter term, defines the amount by which the continuous order-$\infty$ entropy differs from a reference distribution of uniform probability density $k$ over width $k^{-1}$.
Our choice here of
\begin{equation}   \label{eq:HMinDefn}
  {H^\beta_\infty} ( P )  
    \equiv   - \log ( P^{\rm max}  /  \beta )
\end{equation}
 amounts to comparing the entropy of $P(W)$ with that of a uniform distribution over range $\beta^{-1}$ with probability density $\beta$.
We remark this technique is not unique to the order-$\infty$ R\'enyi entropy, but is also necessary if one wishes to arrive at the differential entropy as a limiting case of the discrete Shannon entropy\footnote{For the standard formulation $H(X) = -\int dX P(X) \log P(X)$, dimensions have been ignored, and the implicit reference is a uniform distribution with width $1$ and probability density $1$. The unnormalized expression is  $\lim_{dX\to0} -\sum_X dX P(X) \log \left[ P(X) dX \right]$, which diverges due to the $dX$ in the logarithm.}.

Although any quantity with units of inverse energy could be used as $k$, $\beta$ is a natural choice: 
$-\log \beta $ appears everywhere $\log P^{\rm max} $ appears in our calculations. 
Any alternative choice of normalization $k$ would result in the need for an additional correction term of $\log\left(k/\beta\right)$ in equations~\ref{eq:wEpsBound} and \ref{eq:CrooksWLemma}.
(This extra term can be interpreted as how the entropy of the new reference distribution compares to that of the uniform distribution of range $\beta^{-1}$ and probability $\beta$.)

The bounds in Theorems~\ref{theorem:wDeltaBound}~and~\ref{theorem:WepsBound} shed light on the physical contributions to $W^\varepsilon$. 
To a first approximation, the $\varepsilon$-required work equals $\Delta F$, the work needed to complete the process quasistatically. 
The negative contribution from $\log( 1 - \varepsilon )$ accounts for the tradeoff between work and failure probability: The agent can lower the bound on the required work by accepting a higher failure probability $\varepsilon$.

Outside the quasistatic limit, the system leaves equilibrium, and $W$ fluctuates from trial to trial.
This fluctuation necessitates a protocol-specific correction $H^\beta_\infty( P_\rev )$. 
In cryptography applications, $H_\infty(P)$ quantifies the uniform randomness (and hence resources usable to ensure privacy) extractable from a distribution $P$~\cite{RennerW04}. 
A distribution might have more randomness, but $H_\infty$ quantifies the minimum value (being the lowest-valued R\'enyi entropy).
In our result, ${H^\beta_\infty}(P_\rev)$ can be thought of as the uniform randomness intrinsic to the work distribution.
${H^\beta_\infty}(P_\rev)$ thus quantifies fluctuations in work, caused by irreversibility, that raise the lower bound on $W^\varepsilon$.

Whereas some one-shot results~(e.g.~\cite{BrandaoHNOW15,EgloffDRV12}) involve R\'{e}nyi entropies of states, ${H^\beta_\infty}(P_\fwd )$ is an entropy of a {\em work distribution}. 
${H^\beta_\infty}(P_\fwd)$ captures fluctuation information from all sources that might affect the work distribution. 
These sources include the initial state and the manner in which the protocol is executed (e.g., quickly or quasistatically).
In contrast, an entropy evaluated on states encodes only some of this information. 
By containing an entropy of a work distribution, rather than an entropy of a state, the above results can be applied in a more general setting: they can be related to the output of any procedure, as opposed to the worth of a particular input state under a fixed procedure (usually taken to be the optimal one~\cite{EgloffDRV12}).
Consequently, our results remain independent of the work-extraction model used. 
Theorems~\ref{theorem:wDeltaBound}~and~\ref{theorem:WepsBound} govern (semi)classical and quantum systems, so long the protocol produces a work distribution consistent with Crooks' Theorem.

As theoretical limits, these bounds will always hold true. 
However, to be operationally useful for bounding $w^\delta$ (or $W^\varepsilon$) they must be applied to models where $P^{\rm max}$ can be upper-bounded following enough trials of the experiment. 
This is possible if the distribution $P(W)$ is smooth such that beyond a certain narrowness of bin width, any further division of each bin results in approximately equal probability densities.

Applying information about the protocol executed in one direction, we have used Crooks' Theorem to infer about the opposite direction. 
This information tightens the bound on $W^\varepsilon$ when ${H^\beta_\infty}(P_\rev )\   \geq\ 0$, \footnote{
One-shot entropies of continuous probability density functions are known to assume negative values~\cite{HanelTT09} when they are less entropic than the implicit reference distribution.
}
i.e., when
\begin{equation}   \label{eq:GoodBound}
   P^{\rm max}  <  \beta.
\end{equation} 
Systems described poorly by conventional statistical mechanics tend to satisfy Ineq.~(\ref{eq:GoodBound}). 
Such systems' work distributions have significant spreads relative to the characteristic energy scale $\beta^{-1}$, such that the distribution lacks tall peaks. 
We will present DNA-hairpin experiments as an example.

\subsection{Crooks' Theorem in specific one-shot work-extraction models}
\subsubsection{Tightening a bound in the work-extraction game}
~\\
Egloff \emph{et al.}\ calculate the optimal amount $w^{\delta}_{\rm best}$ of work 
$\delta$-extractable from a state via the most efficient strategy~\cite{EgloffDRV12}. 
Their calculation implies an upper bound on the work extractable via arbitrary strategies.
By applying Theorems~\ref{theorem:wDeltaBound} and~\ref{theorem:WepsBound} to the Egloff \emph{et al.} framework, we can tighten the bound for protocols that satisfy  the assumptions used to derive Crooks' Theorem and for which $P^{\rm max} < \beta$.

In the Egloff \emph{et al.} setting, we can consider a forward protocol that consists of two stages:
First, the thermal state $\gamma_{- \tau} \equiv e^{-\beta H_{-\tau} } / Z_{-\tau}$
transforms into some nonequilibrium state $\sigma$ as the externally driven Hamiltonian changes. 
The system either can remain thermally isolated or can thermalize, provided that the thermalization satisfies detailed balance (see Appendix~\ref{section:ThermModels}). 
The agent can choose one of many possible strategies, e.g., by alternating Hamiltonian changes and thermalizations or by thermally isolating the system throughout the first stage.
Second, $\sigma$ thermalizes to 
$\gamma_\tau \equiv e^{-\beta H_{\tau} } / Z_{\tau}.$ 
This thermalization neither costs nor produces work. 
The entire protocol is encapsulated in
$
   ( \gamma_{ - \tau },   H_{ -\tau } )
   \mapsto
   ( \sigma, H_\tau  )
   \mapsto 
   ( \gamma_\tau,   H_\tau  ).
$
At the start of the reverse protocol, the system begins in the themal state of $H_{\tau}$.
Under the time-reversed process (in which the drive is reversed, such that the Hamiltonian retraces its path through configuration space), the system is transformed into some nonequilibrium state $\sigma'$. Then the state thermalizes to thermal state of $H_{\tau}$.

For protocols which fall into the above category, knowing about one protocol, we can bound the work extractable from, or the work cost of, the opposite protocol:
\begin{corollary}
\label{theorem:CrooksGamewEps2}
The work $\delta$-extractable from each implementation of the reverse protocol, in terms of the forward protocol's ${H^\beta_\infty} (P_\fwd)$, satisfies
\begin{equation}   
   w^{\delta}   
   \leq \frac{1}{ \beta }   \left[
   \log M   \left(   
       \frac{  G^{T} ( \gamma_{\tau} ) }{  1 - \delta  }  \:  
       ||  \:  G^{T} (  \gamma_{-\tau} )  
       \right)
       -   {H^\beta_\infty} (P_{\fwd } )
   \right].
\end{equation}
for $\delta \in [0, 1)$.
\end{corollary}

\begin{corollary}
\label{theorem:CrooksGameWeps}
The work $\varepsilon$-required during each implementation of the forward protocol, in terms of the reverse protocol's $ {H^\beta_\infty (P_\rev)}$, satisfies 
\begin{equation}  \label{eq:CrooksGameWeps0}
   W^\varepsilon  
   \geq\    \frac{1}{ \beta }   \left[
   \log M   \left(   
       \frac{  G^{T} (  \gamma_{-\tau}  ) }{  1 - \varepsilon  }  \:  
       ||  \:  G^{T} (  \gamma_\tau  )  
       \right)
       +    {H^\beta_\infty}  (P_{\rev } )
   \right]
\end{equation}
for $\varepsilon \in [0, 1)$.
\end{corollary}
The proofs appear in Appendix~\ref{app:gametight}.
Each corollary consists of a bound derived from~\cite{EgloffDRV12} and an ${H^\beta_\infty}$ correction attributable to Crooks' Theorem (introduced via Theorems~\ref{theorem:wDeltaBound} and~\ref{theorem:WepsBound}). 
The ${H^\beta_\infty}$ quantifies the protocol's suboptimality, caused by dissipation due to the protocol's speed~\cite{BrowneGDV14}.
Positive values of ${H^\beta_\infty}$ tighten the bounds. 
Hence incorporating information about the forward (reverse) process into the reverse (forward) bound improves the bound when the process deviates sufficiently from the quasistatic ideal.
%

\subsubsection{Modeling fluctuation-relation problems with resource theories}
~\\
\indent One can formulate scenarios governed by Crooks' Theorem in thermodynamic resource theories. Such scenarios involve a sequence of {\em thermal operations} that obey detailed balance. 
Such operations form a strict subset of the set of all thermal operations (see Appendix~\ref{section:ThermModels}). Hence Crooks' Theorem does not govern all thermal operations.
The application of Crooks' Theorem requires the introduction of work and time into the resource theories. Work can be defined in terms of a battery~\cite{SkrzypczykSP14,YungerHalpernR14}; and time, in terms of a clock~\cite{BrandaoHORS13,HorodeckiO13}.
Resource-theory results can be used to derive the work cost of a sequence of transformations that a system governed by Crooks' Theorem can follow (see Appendix~\ref{app:BattAndClock}).

We leave for future work the derivation, from resource-theory results, of testable predictions about Crooks' problem. Considerable mathematical tools, such as monotones~\cite{HorodeckiO13,BrandaoHNOW15,GourMNSYH13} and catalysts~\cite{BrandaoHNOW15,GourMNSYH13}, have been developed within the resource-theory framework.
Having demonstrated the applicability of Crooks' Theorem to resource theories, we look to use Crooks' Theorem to bridge these mathematical tools to experiments.

\section{Examples of one-shot work quantities in fluctuation contexts}

\subsection{Landauer bit reset and Szil\'{a}rd work extraction}

\begin{figure}[htb]
\centering
\includegraphics[width=.65\textwidth, clip=true]{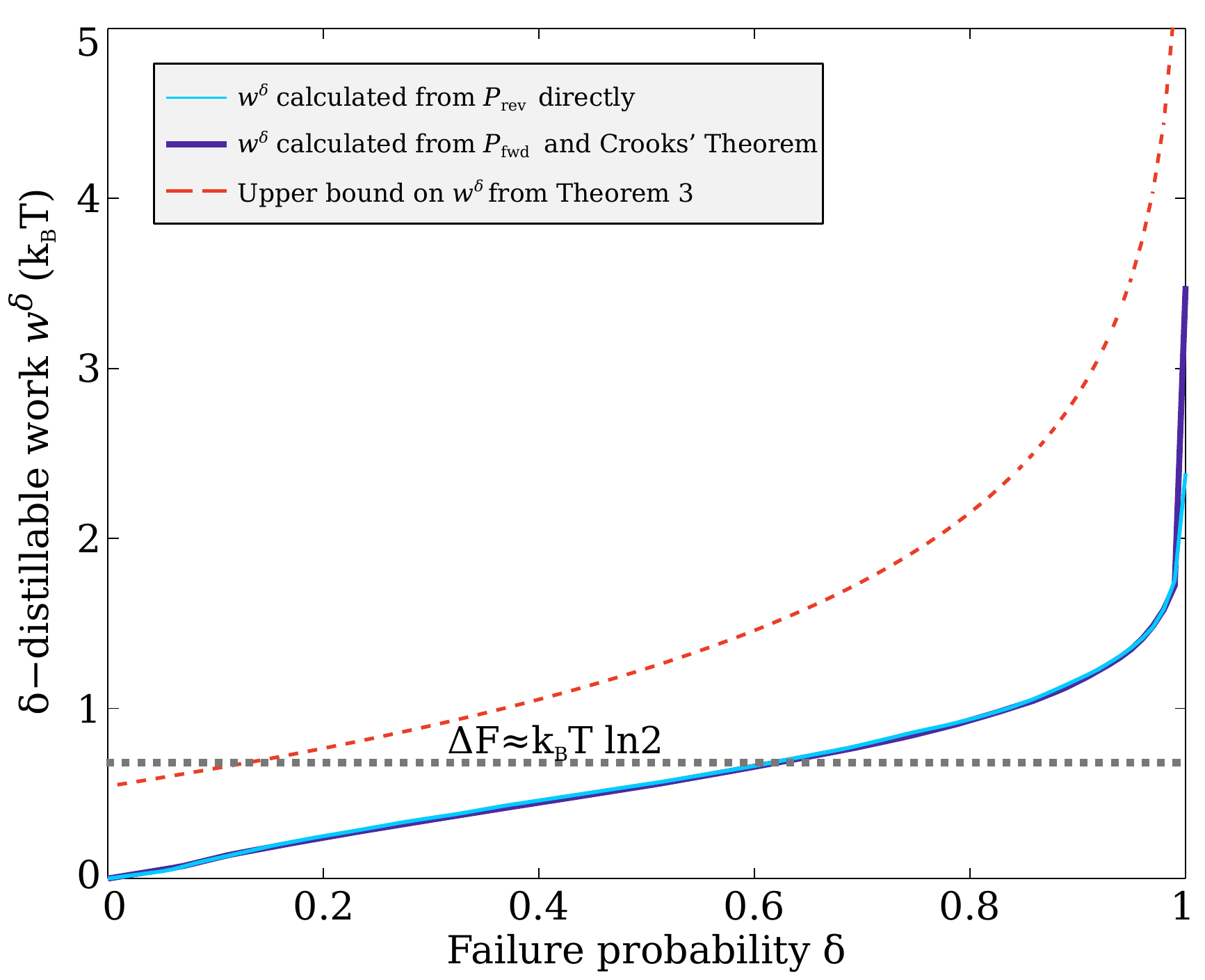}
\caption{
\caphead{One-shot work vs.\ failure probability for numerical simulations of Landauer erasure and Szil\'{a}rd work extraction.}
One-shot work quantities were calculated in a setting governed by fluctuation theorems.
We simulated 10,000 bit-reset trials.
The predicted outputs for work-extraction trials (dark blue), inferred via Crooks' Theorem,
are compared to the work outputs from 10,000 directly simulated work-extraction trials (light blue)
and to the one-shot generalizations of Jarzynski's Equality presented in this article (red dashed).
The horizontal dotted gray line indicates the free-energy difference of the process ($\kB T \ln 2$).
The dark blue curves' coinciding with the light blue curves supports the applicability of Crooks' Theorem to this idealized setting.
The red dashed curves bound the blue curves from above, showing that the one-shot generalizations of Jarzynski's Equality (Theorem~\ref{theorem:wDeltaBound}) governs this scenario.
\label{fig:Simulation}
}
\end{figure}

A simple example involves the heat-exchanging portion of {\em Landauer bit reset} and its reverse, {\em Szil\'{a}rd work extraction}.
The set-up consists of a two-level system $\mathcal{S}$ governed by the Hamiltonian $H( \lambda_t ) = E(t) \ketbra{E}{E} $. 
Suppose $\mathcal{S}$ exchanges energy with a heat bath whose temperature is $T = \frac{1 }{ \kB \beta}$. 
At time $t = -\tau$, $E(t)  =  0$,  and $\mathcal{S}$ is in thermal equilibrium, i.e., in the maximally  mixed state $\rho(-\tau)  =  \frac{1}{2} ( \ketbra{0}{0}  +  \ketbra{E}{E} )$. 
If $\rho$ represents the location of a particle in a two-compartment box, the agent has no idea which compartment the particle occupies. 

Transforming $\rho(-\tau)$ into a pure state---forcing the particle into one half of the box---is called \emph{bit reset}, or \emph{Landauer erasure}. Resetting the bit quasistatically costs, on average,
\begin{equation} \label{eq:LandauerCost}
   \langle W \rangle  =   \int_0^\infty   \frac{   e^{- \beta E }   }{Z}   dE
   = \kB T \log 2.
\end{equation} 
If the bit is reset in a finite time, $\langle W \rangle$ might exceed $\kB T \log 2$~\cite{BrowneGDV14}. 
Such a protocol has appeared in fluctuation contexts before~\cite{EgloffDRV12,BrowneGDV14,Aberg13} and has been realized experimentally (e.g., in a test of Jarzynski's Equality by Brownian motion~\cite{BerutPC13,JunGB14}).

We define failure under the assumption that every started trial is completed: 
a forward trial \emph{fails} if it consumes more work than the budgeted work $W^\varepsilon$.
(Alternatively, one could define ``failure'' under the assumption that too-costly trials would not be completed. 
A trial would be said to fail if the budgeted work were consumed but the bit had not been reset.)

Reversal of the bit reset amounts to \emph{Szil\'{a}rd work extraction}. 
Le\'{o} Szil\'{a}rd envisioned the conversion of information into work in 1929~\cite{Szilard29}. 
$\mathcal{S}$ begins thermally isolated, in the pure state $\ket{0}$, and governed by $H( \lambda_t ) = 0$. 
During the first leg of Szil\'{a}rd work extraction, the agent raises $E(t)$ to infinity. 
The raising costs no work because $\mathcal{S}$ occupies the lower level. 
In the second leg, $\mathcal{S}$ is coupled to the bath, then performs positive work as $E(t)$ decreases to zero.
To be strictly the reverse of the bit reset presented above, we consider only the second leg as an implementation of the reverse protocol.
The initial state, although it is a pure state, is thermal since only the lower energy level has non-zero occupation probability according to the Gibbs distribution.
As such this work-extraction step can be linked to the bit reset via Crooks' theorem; their Hamiltonians are the time-reverse of each other, and each protocol starts in the appropriate thermal state.

\label{section:NumericalSimulation}
These two processes, forming a forward-and-reverse pair governed Crooks' Theorem, provide a natural example with which to test our one-shot results.
We performed a Monte Carlo simulation of Landauer erasure and Szil\'{a}rd work extraction. Details appear  in Appendix~\ref{app:simulation}.
The simulation produced results consistent with Theorems~\ref{theorem:wDeltaBound}, as shown in Fig.~\ref{fig:Simulation}.

\subsection{Experiment: DNA-hairpin unzipping}
When single molecules are manipulated experimentally, ``fluctuations are relevant and deviations from the average behavior are observable''~\cite{AlemanyR10}. 
Some such experiments are known to obey fluctuation relations.
We show that data from DNA-hairpin experiments used previously to test Crooks' Theorem~\cite{MossaMFHR09,ManosasMFHR09,AlemanyR10} agree with the one-shot results in \S\ref{sec:Unification}. 
The agreement suggests that one-shot statistical mechanics might shed light on similar single-molecule experiments and applications. Alternatively, such experiments might be used to test one-shot statistical mechanics.

\begin{figure}
\centering
\includegraphics[width=.8\textwidth, clip=true]{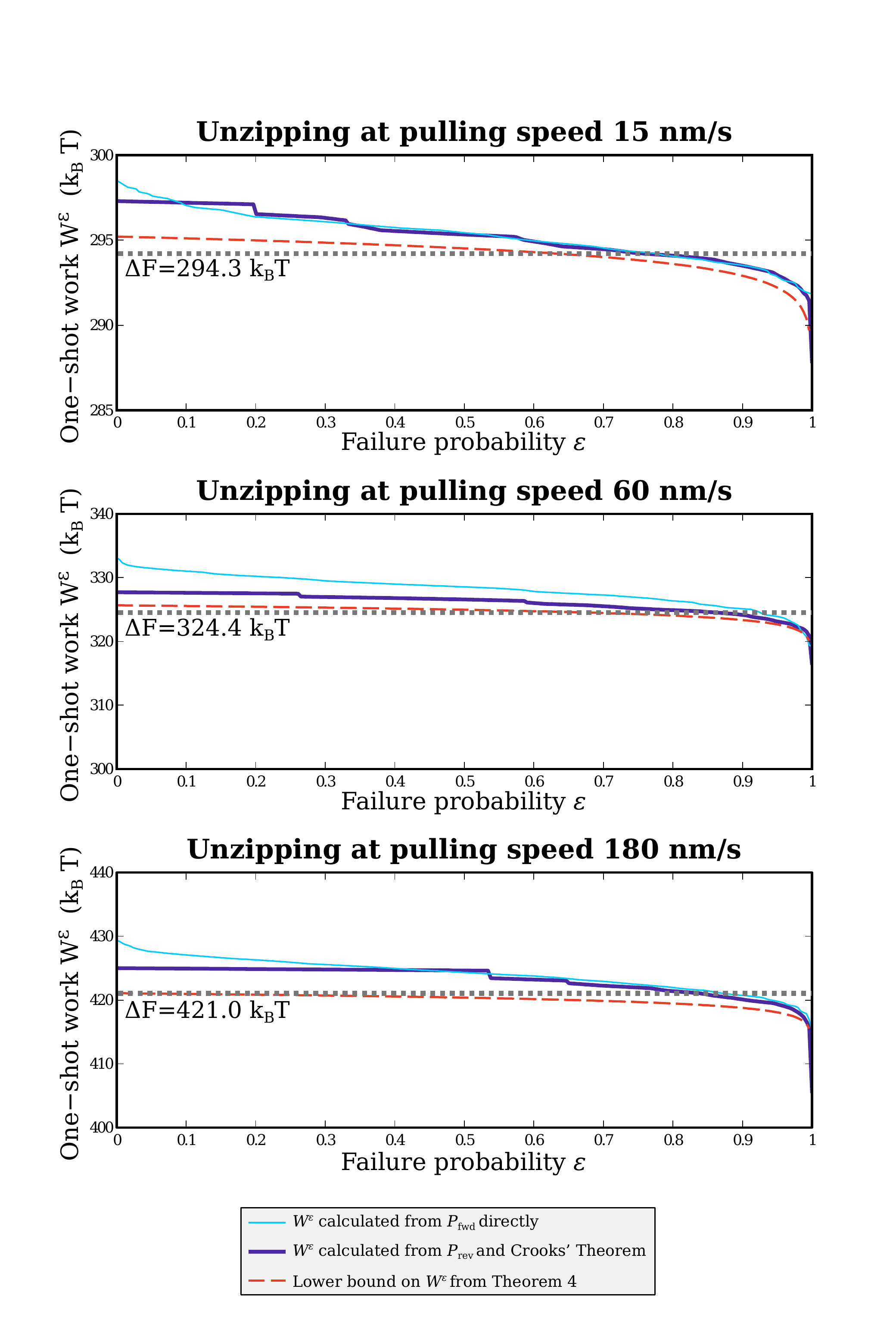}
\caption{\caphead{DNA-hairpin unzipping at three different speeds.}
Using work values from~\cite{MossaMFHR09,AlemanyR10,AlemanyREmails}, we have plotted the one-shot $\varepsilon$-required work $W^\varepsilon$ against the failure probability $\varepsilon$.
The theoretical lower bound (the red dashed line) derived in this article remains close to the experimentally measured value (light blue) and to the value inferred from the reverse protocol's work distribution via Crooks' Theorem (dark blue). 
The gray dotted line indicates the free-energy difference in each process.
The red-blue separation maximizes at about $5\,\kB T$, whereas each trial requires about $300\textendash400\,\kB T$, a relative difference of order $1\%$.
\label{fig:Experiment}
}
\end{figure}

A {\em DNA hairpin} is a short double helix of about 21 base pairs~\cite{MossaMFHR09,ManosasMFHR09,AlemanyR10}. 
The helix's two strands are called \emph{legs}. 
One end of one leg is attached to one end of the other leg, forming a shape like a hairpin's. 
The other end of each leg ends in a \emph{handle} formed from DNA. 
To each handle is attached a polystyrene or silica bead. 
One bead remains anchored on a micropipette. 
The other is caught in an optical trap that exerts a force. 
During the forward protocol, these optical tweezers pull the legs apart, unzipping the DNA into one strand. 
The more quickly the hairpin is split (the greater the \emph{pulling speed}), the more work is dissipated. 
During the reverse protocol, the helix is rezipped. 

We combined work distributions provided by Ritort and Alemany~\cite{MossaMFHR09,AlemanyR10,AlemanyREmails} with Theorem~\ref{theorem:WepsBound} (illustrated in Fig.~\ref{fig:Experiment}).
Each graph shows the work $W^\varepsilon$ that a given unzipping trial is $\varepsilon$-guaranteed to require, plotted against the failure probability $\varepsilon$.
We converted the data (a list of work values) into a probability distribution by forming a histogram with $50$ equally sized bins that span the range of work costs.
Energy is given in units of $\kB T$, such that $\beta=1$. 
For pulling speeds of $15$, $60$, and $\unit[180]{nm\,s^{-1}}$, the binning resulted in distributions whose $P^\mathrm{max} = 0.465 \beta$, $0.277 \beta$, and $0.162 \beta$ respectively.
In all three cases $P^\mathrm{max} < \beta$, such that the $H_{\rm max}$ term in Theorem~\ref{theorem:WepsBound} tightens the bound.

Whereas a work investment of about $300 \textendash 400\,\kB T$ is required to complete the procedure, the jittering between the light blue curve (the directly measured value of $W^\varepsilon$) and the dark blue curve (the value of $W^\varepsilon$ predicted from the reverse work distribution $P_{\rm rev}$ via Lemma~\ref{lemma:WEpsCalcnCrooks})
is of a scale less than $5\,\kB T$. 
Hence Crooks' Theorem interrelates the work probability distributions up to a discrepancy of around $1\%$, as argued in~\cite{MossaMFHR09,AlemanyR10}.
The red curve remains below the dark blue and light blue curves, confirming that the one-shot lower bound in Theorem~\ref{theorem:WepsBound} governs this experimental setting. 
The red curve remains close to---always within about $5\,\kB T$ of (around $1\%$ of $\Delta F$)---the light blue curve that represents directly simulated work investment.
This agreement between theory and experiment suggests that the application of one-shot results may shed light on similar single-molecule experiments and on applications such as molecular motors, thermal ratchets, and nanoscale engines~(e.g.~\cite{LacosteLM08,ChengSHEL12,SerreliLKL07}).

\section{Conclusions and outlook}
Crooks' Theorem relates probability distributions between a process and its reverse. We can manipulate these distributions using tools from one-shot statistical mechanics. 
As demonstrated in this article, combining the toolkits leads to bounds on the work likely to be required (or produced) in classical and quantum process.
Information about fluctuations tighten the bounds.
Fluctuation relations and one-shot statistical mechanics are not competitors, but are mutually compatible. Combining the approaches yields statements about quite general thermal systems.
The combination illustrates a possible bridge from one-shot theory to experimental settings through fluctuation theorems.

One experimental application is the cost of bit reset in modern microprocessors.
As miniaturization reduces the size of transistors further into the nanoscale~(e.g., \cite{MITNano}), limiting only the {\em average} heat dissipation does not ensure that devices work. 
Of increasing importance is a guarantee that no single bit reset dissipates any amount of heat (costs any amount of work) above some threshold that could damage the nanoscale device.
The relevant fluctuations can be studied with Crooks' Theorem and related, via the results in this article, to the one-shot maximum work cost.

The results in this article might be useful also when the work available to be spent on each trial is limited, or if the work extracted from each trial must exceed a certain threshold, except with bounded failure probability.
Such quantities might have uses also in a paranoia setting. 
An agent might have a known amount of work to invest, and one might need to ensure that the agent can not erase some information, except with some small probability. 
Similarly, one-shot work might be applicable in a verification scenario. 
Suppose an agent claims to be able to provide some amount of work.
To test the claim, one can request a transformation that costs more than this amount of work, except in a bounded number of cases.

Future research might reveal further links between one-shot statistical mechanics and fluctuation theorems. 
Here through the analysis of fluctuation theorems,
we identified a relationship between one-shot work quantities and the order-$\infty$ R\'enyi entropy.
By considering the R\'enyi divergence between the work distributions of a process and its reverse, one might find a relationship with one-shot {\em dissipated work}, following from the observation that the average dissipated work is proportional to the Kullback-Leibler divergence (average relative entropy) between the forwards and reverse work distributions~\cite{GomezMarinPvB08}.
Considering divergences between distributions avoids the issue of infinitely high entropy when the work distribution contain sharp peaks that have a finite ratio of height to each corresponding peak in the reverse distribution.
As such, this approach could provide general and robust tools for bridging one-shot statistical mechanics to fluctuation settings that hold for discrete work distributions in addition to the continuous distributions focused on in this article.

{\bf Note added:} Between the first presentation of these results and the current version of this article, related results have appeared in~\cite{DahlstenCBGYHV15,SalekW15}.


\ack
The authors are grateful for conversations with Anna~Alemany, Janet~Anders, Cormac~Browne, Tanapat~Deesuwan, Alex~Lucas, Jonathan~Oppenheim, Felix~Pollock, and Ibon~Santiago. 
This work was supported by a Virginia Gilloon Fellowship, an IQIM Fellowship, support from NSF grant PHY-0803371, the FQXi Large Grant for ``Time and the Structure of Quantum Theory,'' the~EPSRC, the John Templeton Foundation, the Leverhulme Trust, the Oxford Martin School, the National Research Foundation (Singapore), and the Ministry of Education (Singapore). 
The Institute for Quantum Information and Matter (IQIM) is an NSF Physics Frontiers Center with support from the Gordon and Betty Moore Foundation.
VV and OD acknowledge funding from the EU Collaborative
Project TherMiQ (Grant Agreement 618074).
NYH was visiting the University of Oxford under the auspices of Jonathan Barrett and the Department of Atomic and Laser Physics while much of this paper was developed.

\newcommand{\newblock}{}

\markboth{\ourtitle}{\ourtitle}

\newpage
\appendices
\section{Relationships among thermalization models}
\label{section:ThermModelsMain}
\label{section:ThermModels}

\begin{figure}[htb]
\centering
\includegraphics[width=0.65\textwidth]{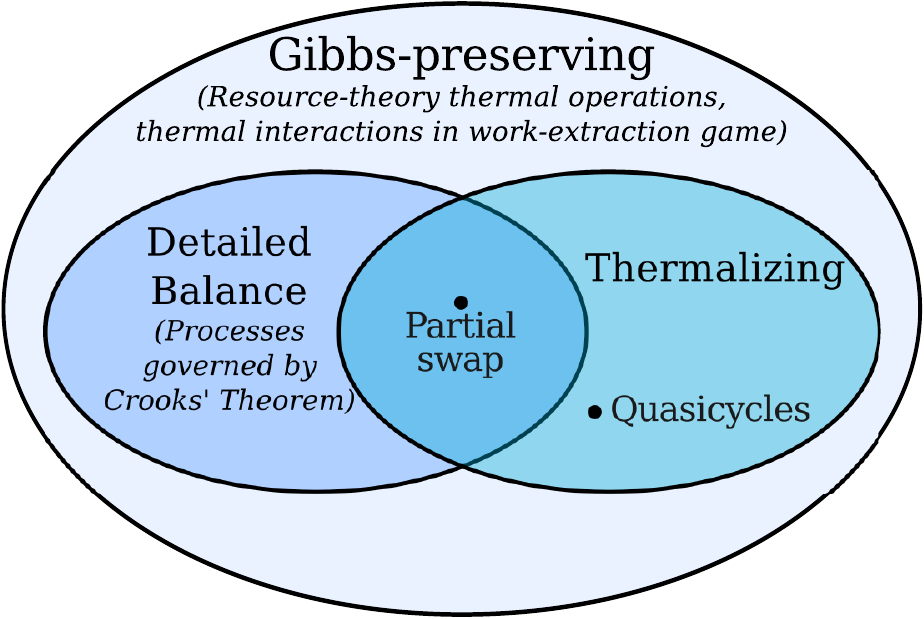}
\caption{\label{fig:ThermalVenn} Venn diagram illustrating the relationships among properties of stochastic matrices that model thermal interactions. Dots represent example models.}
\end{figure}

Consider a system $\mathcal{S}$ governed by a discrete $N$-level Hamiltonian $H$. Suppose that $\mathcal{S}$ interacts with a heat bath whose inverse temperature is $\beta$. An $N$-dimensional probability vector $\vec{s}$ represents the system's state. Whole or partial thermalization of $\mathcal{S}$ can be modeled as a sequence of discrete steps, each represented by a stochastic matrix. Different possible properties of such matrices characterize different models of heat exchanges. We address the properties of Gibbs-preservation, detailed balance, and thermalization. By $\vec{g}$, we denote the probability vector that represents the Gibbs state associated with $H$ and $\beta$:
\begin{equation*}
   \vec{g}   =
   \left( \frac{  e^{ - \beta E_1 }  }{Z},   \frac{  e^{ - \beta E_2 }  }{Z},   \ldots,   
            \frac{  e^{ - \beta E_N }  }{Z}  \right).
\end{equation*}

A matrix $M$ is \emph{Gibbs-preserving} relative to $H$ and $\beta$ if $M$ maps the corresponding Gibbs state to itself:
\begin{equation} \label{eq:GibbsPreserving}
   M \vec{g}   =   \vec{g}.
\end{equation}
Gibbs preservation constrains the unit-eigenvalue eigenspace of $M$. 
The set $\mathcal{G}$ of Gibbs-preserving matrices is equivalent to the set of thermal operations~\cite{HorodeckiO13} and to the set of thermal interactions in the game~\cite{EgloffDRV12}.

A strict subset of $\mathcal{G}$ is the set $\mathcal{D}$ of detailed-balanced matrices:   $\mathcal{D}  \subset  \mathcal{G}$. 
Let $A$ and $B$ denote microstates associated with the energies $E_A$ and $E_B$. $M$ encodes the probabilities that $\mathcal{S}$ transitions from $A$ to $B$, and vice versa, during one heat-exchange step. If these probabilities satisfy
\begin{equation} 
   \label{eq:DetailedBalance}
P(A  \mapsto  B) 
   =   
P(B \mapsto A)
e^{ - \beta ( E_B  -  E_A ) },
\end{equation}
$M$ obeys \emph{detailed balance}~\cite{Crooks99}. 

If the steps in an extended heat exchange obey detailed balance, the extended heat exchange obeys \emph{microscopic reversibility}~\cite{Crooks98}. 
From the assumption that heat exchanges are microscopically reversible, Crooks derives his theorem~\cite{Crooks99}.
Hence if the heat exchanges in a process obey detailed balance (and the other assumptions used to derive Crooks' Theorem, such as initialization to a thermal state), the process obeys Crooks' Theorem.

Crooks defines microscopic reversibility as follows while deriving his theorem~\cite{Crooks99}. Let $P( x(t) | \lambda_t )$ denote the probability that, if the external parameter varies as $\lambda_t$ during some forward trial, the state of the classical system $\mathcal{S}$ follows the phase-space trajectory $x(t)$. 
The ``corresponding time reversed path'' is denoted by
$( \bar{\lambda}( -t ),  \bar{x}(-t) )$.
Let the functional $Q[ x(t), \lambda_t ]$ denote the heat that $\mathcal{S}$ ejects if $\lambda_t$ and $x(t)$ characterize the forward trial. The heat exchange obeys \emph{microscopic reversibility} if
\begin{equation}
   \label{eq:MicroRevDefn}
   \frac{ P( x(t) | \lambda_t ) }{ P( \bar{x}(-t) | \bar{\lambda}(-t) ) }
   =   e^{ - \beta Q[ x(t), \lambda_t ] }.
\end{equation}

Another strict subset of Gibbs-preserving matrices is the set $\mathcal{T}$ of thermalizing matrices: $\mathcal{T} \subset \mathcal{G}$. 
We call a matrix $M$ \emph{thermalizing} if it evolves every state $\vec{s}$ of $\mathcal{S}$ toward the Gibbs state associated with $H$ and $\beta$:
\begin{equation} \label{eq:ThrmLimitMain}
   \lim_{n\to\infty} M^n \vec{s} = \vec{g}.
\end{equation}
Equation~\eref{eq:ThrmLimitMain} encapsulates intuitions about what ``thermalization'' means. Some matrices that model thermal interactions in the game and in the resource theories violate Eq.~\eref{eq:ThrmLimitMain}, as do some thermal interactions governed by Crooks' Theorem. $\mathcal{T}$ overlaps with $\mathcal{D}$.

The properties we have introduced---Gibbs preservation, detailed balance, and thermalization---imply relationships among Crooks' Theorem, theorems about the game, and resource-theory theorems. 
The game, as well as the resource theories, model the heat exchanges in some processes governed by Crooks' Theorem. 
Crooks' Theorem does not necessarily govern all heat exchanges possible in the game or in the resource theories.

\subsection{Proof of Venn diagram}

Let us justify our modeling of processes governed by Crooks' Theorem with the work-extraction game in~\cite{EgloffDRV12} and with resource theories. Different \emph{frameworks} (Crooks' theorem, the game, and the resource theories) model interactions with heat baths differently. One step in an interaction can be represented by a stochastic matrix that has at least one of three properties: Gibbs-preservation ($\mathcal{G}$), detailed balance ($\mathcal{D}$), and thermalization ($\mathcal{T}$). The relationships among these matrices are summarized in figure~\ref{fig:ThermalVenn} and in the following statements:
\begin{enumerate}

\item \makebox{$\mathcal{T} \subset \mathcal{G}$:} 
All thermalizing matrices are Gibbs-preserving (Lemma~\ref{lem:th_implies_gp}), but not vice versa (Lemma~\ref{lem:gp_not_implies_th}).

\item \makebox{$\mathcal{D} \subset \mathcal{G}$:} 
All detailed-balanced matrices are Gibbs-preserving (Lemma~\ref{lem:db_implies_gp}), but not vice versa (Lemma~\ref{lem:gp_not_implies_db}).

\item\makebox{$\mathcal{D} \neq \mathcal{T}$, $\mathcal{D} \not\subset \mathcal{T}$, and $\mathcal{T} \not\subset \mathcal{D}$:} 
Obeying detailed balance is not equivalent to being thermalizing, and neither category is a subset of the other.
(Lemma~\ref{lem:db_is_not_th}).

\item \makebox{$\mathcal{D} \cap \mathcal{T} \neq \emptyset$:} 
Some matrices are detailed-balanced and themalizing (Lemma~\ref{lem:db_th_exist}).
\end{enumerate}
While proving these claims, we justify the inclusion of two example matrices, the partial swap and quasicycles, in figure~\ref{fig:ThermalVenn}.

The proofs contain the following notation: $\mathcal{S}$ denotes a quasiclassical system that evolves under a Hamiltonian $H$ and that exchanges heat with a bath whose inverse temperature is $\beta$. By $\vec{s}  =  (s_1, s_2, \ldots, s_d)$, we denote the state of $\mathcal{S}$. The vector's elements are the diagonal elements of a density matrix relative the eigenbasis of $H$. The Gibbs state relative to $H$ and to $\beta$ is denoted by $\vec{g}$.

To prove some of the foregoing claims, we characterize thermalizing matrices with the Perron-Frobenius Theorem~\cite{HornJ85}. The theorem governs irreducible aperiodic nonnegative matrices $M$.\footnote{
By \emph{nonnegative,} we mean that every element of $M$ is no less than zero.}
Consider the eigenvalue $\lambda$ of $M$ that has the greatest absolute value. 
According to the Perron-Frobenius Theorem, $\lambda$ is the only positive real eigenvalue of $M$, and $\lambda$ is associated with the only nonnegative eigenvector $\vec{v}_\lambda$ of $M$.
Suppose that $M$ is stochastic, such that $\lambda = 1$. 
By the spectral decomposition theorem, $\lim_{n \to \infty} M^n \vec{s}  =  \vec{v}_\lambda$. 
If $\vec{v}_\lambda = \vec{g}$, the matrix is thermalizing.

%
%
\begin{lemma}   \label{lem:th_implies_gp}
All thermalizing matrices are Gibbs-preserving.

\begin{proof}
Let $M$ denote a thermalizing matrix associated with the same Hamiltonian and $\beta$ as $\vec{g}$. For all states $\vec{s}$ of $\mathcal{S}$,
\begin{equation}   \label{eq:ThrmLimit}
   \lim_{n\to\infty} M^n \vec{s} = \vec{g}.
\end{equation}
To prove the lemma by contradiction, we suppose that $M$ does not map $\vec{g}$  to itself:  $\vec{g} \not\mapsto  \vec{g}$.

Premultiplying Eq.~\eref{eq:ThrmLimit} by $M$ generates 
\begin{equation}   \label{eq:ThrmLimit2}
   \lim_{n\to\infty} M M^n \vec{s} = M \vec{g} \neq \vec{g}.
\end{equation}
This equation contradicts 
\begin{equation}
   \lim_{n\to\infty} M M^n   \vec{s} 
   = \lim_{n\to\infty} M^{n+1}   \vec{s} 
   = \lim_{n\to\infty} M^{n}   \vec{s} 
   =   \vec{g}.
\end{equation}
By the contrapositive, all thermalizing matrices are Gibbs-preserving.
\end{proof}
\end{lemma}

%
%
\begin{lemma}   \label{lem:gp_not_implies_th}
Not all Gibbs-preserving matrices are thermalizing.
\begin{proof}
In general, this will be the case for matrices which have the Gibbs state as an eigenvector but do not otherwise satisfy the conditions of irreducibility or aperiodicity required for the Perron--Frobenius theorem to govern their behavior.
To prove the lemma by example, we construct one Gibbs-preserving matrix that is not thermalizing.
Consider a block-diagonal stochastic $N \times N$ matrix $M$. (Being block-diagonal, $M$ is reducible).
Let $M$ decompose into two submatrices: $M  = M_1  \oplus  M_2$. Let $M_1$ be defined on the first $n_1$ energy levels, and let $M_2$ be defined on the remaining $n_2$ energy levels.

Denote by $\vec{g_1}$ the Gibbs state associated with the first $n_1$ energies (and the partition function $Z_1$), and by $g_2$ the Gibbs state associated with the final $n_2$ energies (and the partition function $Z_2$). Suppose that 
\begin{equation}
   \tilde{g}_1   \equiv   g_1 \oplus (\underbrace{0, 0, \ldots, 0}_{n_2}) 
   \quad {\rm and} \quad 
   \tilde{g}_2   \equiv   
   (\underbrace{0, 0, \ldots, 0}_{n_1}) \oplus g_2
\end{equation}
are normalized probability eigenvectors of $M$, each associated with the unit eigenvalue. 
Every vector of the form
\begin{equation}
   \vec{\nu}_\alpha = \alpha \tilde{g}_1 + (1-\alpha) \tilde{g}_2
\end{equation}
is also a normalized probability eigenvector of $M$ associated with the unit eigenvalue. 

The possible forms of $\vec{\nu}_\alpha$ form a family. One member of the family is the Gibbs state $\vec{g}$, which corresponds to $\alpha = Z_1 / Z$ (wherein $Z$ denotes the total partition function). 
Hence $\vec{g} \in \{\vec{\nu}_\alpha\}_{\alpha\in[0,1]}$ is an eigenvector of $M$, and $M$ is Gibbs-preserving. 
However, $\vec{g}$ is not the only eigenvector associated with the unit eigenvalue. $M$ does not evolve every initial state toward $\vec{g}$. 
In general, $\lim_{n\to\infty} M^n \vec{s} = \nu_\alpha$,
wherein $\alpha$ is the total occupation probability of the first $n_1$ energy levels of $\vec{s}$. 
As some states $\vec{s}$ correspond to $\alpha\neq Z_1/Z$ and to $\nu_\alpha\neq g$, $M$ does not map every initial state to the Gibbs state. Hence $M$ is not thermalizing.
Our claim has been proved by example.

\end{proof}
\end{lemma}
\noindent Together, Lemmas~\ref{lem:th_implies_gp} and~\ref{lem:gp_not_implies_th} imply the strict relation $\mathcal{T} \subset \mathcal{G}$.

%
%
\begin{lemma}   \label{lem:db_implies_gp}
All matrices that obey detailed balance relative to the Hamiltonian $H$ and the inverse temperature $\beta$ preserve the Gibbs state $\vec{g}$ associated with $H$ and $\beta$.

\begin{proof}
We will write the forms of the elements in an arbitrary detailed-balanced stochastic $N \times N$ matrix $M$. By performing matrix multiplication explicitly, we show that $M \vec{g} =\vec{g}$.

Let $M_{ij}$ denote the element in the $i^{\rm th}$ row and $j^{\rm th}$ column of $M$. This element equals the probability that, upon beginning in the $j^{\rm th}$ energy level, a system $\mathcal{S}$ transitions to the $i^{\rm th}$ level. Let $g_i$ denote the thermal population of level $i$ (the $i^{\rm th}$ element of $\vec{g}$).

Detailed balance and stochasticity constrain the relationships among the $M_{ij}$. By the definition of detailed balance [Eq.~\eref{eq:DetailedBalance}], the elements in the lower left-hand triangle of $M$ are related to the elements in the upper right-hand triangle by
\begin{equation}   \label{eq:DBImp}
M_{ji} = M_{ij} \frac{ g_j }{ g_i } \qquad\forall~j > i.
\end{equation}

The matrix has the form
\begin{equation}
\label{eq:DB_SM}
   M =  \left(\begin{array}{cccc}
           M_{11} & M_{12} & M_{13}& \ldots \\
           M_{12} \frac{g_2}{g_1} & M_{22} & M_{23} & \ldots \\
           M_{13} \frac{g_3}{g_1} & M_{23} \frac{g_3}{g_2} & M_{33} & \ldots \\
           \vdots& \vdots & \vdots & \ddots 
           \end{array} \right).
\end{equation}

Because $M$ is stochastic, the elements in each column sum to one. This normalization condition fixes each diagonal element $M_{ii}$ as a function of the other $M_{ij}$ that occupy the same column:
\begin{eqnarray}
\label{eq:DB_SM_Norm}
M_{ii} 	& = & 1 - \sum_{j<i} M_{ji} - \sum_{j>i} M_{ij}\frac{g_j}{g_i}. 
\end{eqnarray}

Using index notation, we ascertain how $M$ transforms the Gibbs state $\vec{g}$:
\begin{eqnarray}
\sum_{j} M_{ij} g_j     
         & = & M_{ii} g_i + \sum_{j<i} M_{ij} g_j  + \sum_{j>i} M_{ij} g_j \nonumber \\
         & = & g_i - \sum_{j<i} M_{ji} g_i - \sum_{j>i} M_{ij}\frac{g_j}{g_i} g_i
				+ \sum_{j<i} M_{ij} g_j  + \sum_{j>i} M_{ij} g_j \nonumber \\
         & = & g_i - \sum_{j<i} M_{ij} g_j - \sum_{j>i} M_{ij}g_j
				+ \sum_{j<i} M_{ij} g_j  + \sum_{j>i} M_{ij} g_j \nonumber \\
		 & = & g_i
\end{eqnarray}
The second line follows from the substitution of Eq.~\eref{eq:DB_SM_Norm} for $M_{ii}$.
The third line follows from the substitution of Eq.~\eref{eq:DBImp} into the elements of first sum.

We have shown that $\vec{g}$ is an eigenvector of $M$ that corresponds to the unit eigenvalue. 
An $N \times N$ matrix $M$ that obeys detailed balance relative to $H$ and $\beta$ preserves the Gibbs state associated with $H$ and $\beta$.

\end{proof}
\end{lemma}

%
%
\begin{lemma}   \label{lem:gp_not_implies_db}
Not every Gibbs-preserving matrix for some Hamiltonian $H$ and inverse temperature $\beta$ satisfies detailed balance for $H$ and $\beta$.

\begin{proof}
\begin{figure}[htb]
\centering
\includegraphics[width=0.35\textwidth]{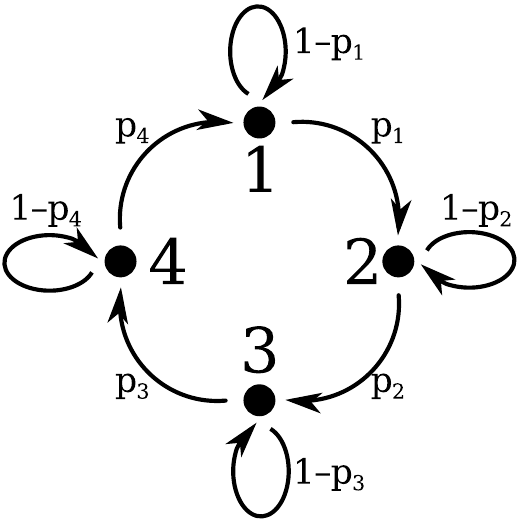}
\caption{\label{fig:Quasicycle} \caphead{Directed graph that illustrates a four-level quasicycle:} The associated matrix fails to satisfy detailed balance, but cunning choice of the $p_i$ ensures that the matrix is thermalizing.}
\end{figure}

Gibbs-preservation only places a restriction on one eigenvalue of a matrix, and so there remains enough freedom to choose a matrix exhibiting this property, but not detailed balance.
An example of such is a quasicycle, described in~\cite{HorodeckiO13}. 
A {\em quasicycle} is a process that has a probability $P(i   \mapsto   (i+1)  \mod{N})  \equiv  p_i$ of evolving a system $\mathcal{S}$ that occupies energy eigenstate $i$ to eigenstate $i+1$ and has a probability $1 - p_i$ of keeping $\mathcal{S}$ in state $i$. 
All $p_i>0$, and for at least one value of $i$, $p_i < 1$.
The directed graph of a quasicycle forms a ring in which at least one node also has a loop to itself, corresponding to a value of $(1-p_i)>0$. An example appears in figure~\ref{fig:Quasicycle}.
The probability that $i$ evolves to $j$ forms element $M_{ij}$ of matrix $M$. The matrix fails to satisfy detailed balance if $\mathcal{S}$ has more than three energy eigenstates, because ${P(i \mapsto (i+1) \mod{N}} )$ is finite, though 
${P((i+1)~\mod{N} \mapsto i)} = 0$.

We will show that, if the $p_i$ assume certain values, the Gibbs state $\vec{g}$ is an eigenvector of $M$. Let level $i = 1$ correspond to the lowest energy eigenvalue. Solutions of the form
\begin{equation}   \label{eq:ChosenP}
p_i =  p \: \frac{g_1}{g_i},
\end{equation}
wherein $p \in (0, 1)$ denotes a free parameter in the range $(0,1)$ are Gibbs-preserving.
Roughly, the greater the value of $p$, the more quickly the quasicycle is traversed.

To verify that Eq.~\eref{eq:ChosenP} describes a Gibbs-preserving matrix, we express the matrix multiplication $M \vec{g}$ in index form:
\begin{eqnarray}
   (M\vec{g})_1 & = & (1 - p_1)  g_1 + p_N   g_N   \\
   (M\vec{g})_i & = &  (1 - p_i)  g_i + p_{i-1}   g_{i-1} 
   \qquad i = 2\ldots N.
\end{eqnarray}
Upon substituting in from Eq.~\eref{eq:ChosenP}, we can simplify the equations to
\begin{eqnarray}
   (M\vec{g})_1
           & = &   g_1 \\
   (M\vec{g})_i
          & = &  g_i       \qquad i = 2\ldots N.
\end{eqnarray}
Hence $M \vec{g}  =  \vec{g}$, so $M$ preserves Gibbs states.
\end{proof}
\end{lemma}
\noindent Together, Lemmas~\ref{lem:db_implies_gp} and~\ref{lem:gp_not_implies_db} imply the strict relation $\mathcal{D} \subset \mathcal{G}$.

%
%
\begin{lemma}   \label{lem:db_is_not_th}
Obeying detailed balance is not equivalent to thermalizing: $\mathcal{D} \neq \mathcal{T}$, nor is one category a subset of the other.

\begin{proof}
We will show that the quasicycle matrix $M$ described in the proof of Lemma~\ref{lem:gp_not_implies_db}---a matrix that does not obey detailed balance---is thermalizing.
Because $M$ is stochastic by construction, its greatest eigenvalue equals one.

In addition to being stochastic, $M$ is irreducible, aperiodic\footnote{
Though quasicycles look cyclic, they are aperiodic. 
For the purposes of the Perron-Frobenius Theorem, a matrix's period is the maximum value $k_{\rm max}$ of $k$ that satisfies the statement ``A system prepared in level $A$ has a nonzero probability of evolving to level $A$ only (but not necessarily) after multiples of $k$ steps''. For irreducible matrices, the possible values of $k$ do not depend on the choice of $A$~\cite{HornJ85}. When this index $k_\mathrm{max}=1$, the matrix is {\em aperiodic}.
Every quasicycle contains at least one node that transitions to itself [not all $p_i=1$ so at least one loop satisfies $P(i \mapsto i) = \left(1-p_i\right) > 0$]. Thus any value of $k$ satisfies the above statement.
Hence $k_{\rm max} = 1$, and quasicycles are aperiodic.
}
and non-negative. 
By the Perron-Frobenius Theorem, the greatest eigenvalue $\lambda$ of $M$ corresponds to the only nonnegative eigenvector $\vec{v}_\lambda$ of $M$.
This $\lambda = 1$, because $M$ is stochastic. As shown in the proof of Lemma~\ref{lem:gp_not_implies_db}, $\vec{v}_\lambda = \vec{g}$. As explained below the proof of Lemma~\ref{lem:th_implies_gp}, $\lim_{n \to \infty} M^n$ maps every vector $\vec{s}$ to $\vec{g}$: The matrices that represent quasicycles thermalize. $M$ does not obey detailed balance, as discussed in the proof of Lemma~\ref{lem:db_implies_gp}.

A simple example of a matrix that obeys detailed balance, but is not thermalizing, is the identity matrix.
Less trivially, it is possible in general to engineer a block-diagonal matrix of the form given in Lemma~\ref{lem:gp_not_implies_th}, and if each block obeys detailed balance, the matrix as a whole will also obey detailed balance [noting that $P(A\!\mapsto\!B) = P(B\!\mapsto\!A) = 0$ trivially satisfies detailed balance]. Such a matrix will not be thermalizing.
Hence we see that thermalizing is not equivalent to obeying detailed balance: $\mathcal{D} \neq \mathcal{T}$, and neither category is a subset of the other.

\end{proof}
\end{lemma}

%
%
\begin{lemma}   \label{lem:db_th_exist}
Some matrices are detailed-balanced and thermalizing: 
$\mathcal{D} \cap \mathcal{T} \neq \emptyset$.

\begin{proof}
We can prove this lemma by example. After reviewing the form of the partial-swap matrix $M$, we show that $M$ thermalizes, then show that $M$ obeys detailed balance.

Partial swap was introduced in Appendix~\ref{app:simulation}. A partial-swap operation has some probability $p$ of replacing the operated-on state $\vec{s}$ with a thermal state and a probability $1 - p$ of preserving $\vec{s}$. If $\vec{s}$ denotes the state of an $N$-level system,
\begin{equation} 
\label{eq:partial_swap}
   M  = \left(1-p\right) \id_N + p G,
\end{equation}
wherein $\id_N$ denotes the $N \times N$ identity and every column of the matrix $G$ is the thermal state $\vec{g}$.

Let us prove that $M$ thermalizes. $M$ is stochastic, as it is the probabilistic combination of $\id_N$ and $G$, which are stochastic. If $N$ is finite, $G$ is positive; so when $p > 0$, $M$ is positive. Positivity implies irreducibility and aperiodicity.
Hence the Perron-Frobenius Theorem\footnote{
For strictly positive matrices, the earlier Perron Theorem implies the same result.}
implies that $M$ has just one nonnegative eigenvector $\vec{v}_\lambda$ and that this eigenvector corresponds to $\lambda=1$.
Direct multiplication shows $M\vec{g}  =  \vec{g}$. Thus, $\vec{g} = \vec{v}_\lambda$ is the only nonnegative eigenvector of $M$ and corresponds to the largest eigenvalue. 
By the argument above Lemma~\ref{lem:th_implies_gp}, $M$ thermalizes.

To show that $M$ obeys detailed balance, we compare the matrix elements that represent the probabilities of transitions between states $i$ and $j$:
\begin{eqnarray}
   \frac{ P( i \mapsto j ) }{ P( j \mapsto i ) }
   & =   \frac{ M_{ji} }{ M_{ij} } \nonumber \\
   & =   \frac{ (1 - p) \delta_{ij}   +   p g_j }{ (1 - p) \delta_{ij}  +  p g_i }  \nonumber \\
   & =   e^{ - \beta (E_j  -  E_i ) },
\end{eqnarray}
wherein $\delta_{ij}$ denotes the Kronecker delta. This equation recapitulates the definition of detailed balance [Eq.~\eref{eq:DetailedBalance}]. Hence matrices---such as the partial swap---can  obey detailed balance while thermalizing.
\end{proof}
\end{lemma}

\section{Quantum derivation of generalized Jarzynski equalities}
\label{app:quantJE}
The results in \S\ref{sec:OneShotCFT} apply to classical and quantum systems. 
To shed extra light on quantum applications, we present an alternative derivation of Lemma~\ref{lemma:WEpsCalcnCrooks} for a quantum system whose energy spectrum is discrete and that lacks contact with the heat bath while its Hamiltonian changes. 

Work is defined as the difference between the outcomes of energy measurements near the protocol's start and end. This definition of work, which appears in~\cite{EngelN07,DornerCHFGV13,HideV10}, differs from the definition in~\cite{TalknerLH07}. 
The discrete version of $\chi^\varepsilon_{\rm rev}(\beta)$ will be defined via analogy with Eq.~\eref{eq:ChiEpsRev}:
\begin{equation}
   \chi_{\rm rev}^\varepsilon ( \beta )
   \equiv   \sum_{W  \geq  - W^\varepsilon}   P_\rev(W)  e^{- \beta W}.
\end{equation}

Let $\mathcal{S}$ denote a quantum system characterized by an external parameter $\lambda_t$ and governed by a Hamiltonian $H( \lambda_t )$ whose energy spectrum is discrete. 
Let $\beta$ denote the inverse temperature of the heat bath with which $\mathcal{S}$ interacts at times $t \in ( - \infty, - \tau)$.
At $t = - \tau$, $\mathcal{S}$ is projectively measured in the energy eigenbasis, then isolated from the bath. Until $t = \tau$, a unitary $U( 2\tau )$ evolves $H (\lambda_t)$ to $H_\tau$, and $\mathcal{S}$ is perturbed out of equilibrium. At $t = \tau$, the energy of $\mathcal{S}$ is measured projectively. Define the work $W$ performed on $\mathcal{S}$ as the difference between the measurements' outcomes. 

\begin{lemma*}
The $\varepsilon$-required work satisfies
\begin{equation}   \label{eq:WEpsCalcnCrooks3}
   \chi^{ \varepsilon }_\rev ( \beta )
   =   (1  -  \varepsilon)   e^{ \beta \Delta F }
   \quad \forall \varepsilon \in [0, 1].
\end{equation}
\end{lemma*}

\begin{proof}
Let $\{   \ket{   \phi_m ( - \tau )   }   \}$ and $\{   E_m ( - \tau)   \}$   denote the eigenstates and eigenvalues of $H( \lambda_{ -\tau }   )$, and let $\{   \ket{   \phi_n (  \tau )   }   \}$   and   $\{   E_n ( \tau)   \}$   denote those of $H( \lambda_\tau   )$. If the measurements yield outcomes $m$ and $n$, the forward trial consumes   $W   \equiv   E_n( \tau )   -   E_m( - \tau )$.

The time-reversed protocol proceeds from $t = \infty$ to $t = -\infty$ and is defined as in the introduction.
Let $p_n ( \tau )$ denote the probability that the first measurement during a reverse trial yields $E_n (\tau)$; and let $p_\rev (m | n)$ denote the probability that, if the first measurement yields $E_n (\tau)$, the second yields $E_m( - \tau)$. By definition,
\begin{eqnarray}
\hspace{-2.5em}
   \chi^\varepsilon_\rev (\beta)   
    =
   \sum_{m, n}  p_\rev (m | n )   p_n ( \tau)   
   e^{ - \beta [  E_m ( - \tau )   -   E_n ( \tau )  ]   }   \label{eq:QDerivn} 
   \;   \Theta(   W^\varepsilon   -   E_n (  \tau )   +   E_m (  -\tau )    ),  
\end{eqnarray}
wherein
\begin{equation}   \label{eq:Theta} 
   \Theta (   W^\varepsilon   -   W_0   )   \equiv  
   \left\{ \begin{array}{lcr}
      \hfill 1 & \mathrm{if} & W^\varepsilon   \geq   W_0 \\
      \hfill 0 & otherwise. &
   \end{array}\right.
\end{equation}

Invoking $p_n ( \tau )  =  e^{ - \beta E_n ( \tau ) } / Z_{ \tau }  $, we cancel the $E_n ( \tau )$-dependent exponentials.
$p_\rev ( m | n )$ equals the  probability $p_\fwd ( n | m )$ that, if an energy measurement at 
$t = - \tau$ yields $E_m (- \tau )$ during a forward trial, an energy measurement at $\tau$ yields $E_n ( \tau )$:
\vbox{
\begin{eqnarray} 
   p_\rev ( m | n ) 
   & =   
   {\rm Tr} ( 
   \ketbra{ \phi_m ( -\tau ) }{ \phi_m ( -\tau ) }   \:  U^\dag ( 2 \tau )\:     
    \ketbra{ \phi_n (  \tau ) }{ \phi_n (  \tau ) }   \:  U( 2 \tau )      )   \nonumber\\
   & =
   {\rm Tr} ( \ketbra{ \phi_n (  \tau ) }{ \phi_n (  \tau ) }   \:  U( 2 \tau )\:     
   \ketbra{ \phi_m ( -\tau ) }{ \phi_m ( -\tau ) }   \:  U^\dag ( 2 \tau )      )  \nonumber \\
   & = 
   p_\fwd ( m | n ).
\end{eqnarray}
}
Substitution into Eq. (\ref{eq:QDerivn}) yields
\begin{eqnarray}
   \chi^\varepsilon_\rev (\beta)   
   =   \frac{1}{   Z_\tau   }
   \sum_{m, n}   p_\fwd (n | m )
   e^{   - \beta   E_m   ( - \tau)   }  
   \; \Theta   (    W^\varepsilon   -   E_n (  \tau )   +   E_m (  -\tau )   ).
\end{eqnarray}
Upon multiplying by   $Z_{ - \tau }    /   Z_{ - \tau } $, we replace 
$e^{   -   \beta    E_m ( - \tau )   }   /   Z_{   -\tau   }$   with   $p_m ( - \tau )$:
\begin{eqnarray}
   \chi^{\varepsilon}_\rev (\beta)
   & = 
   \frac{  Z_{ - \tau }  }{ Z_\tau  }
   \sum_{m, n}   p_m ( - \tau )   p_\fwd ( n | m )  
   \: \Theta (   W^\varepsilon   -   E_n (  \tau )   +   E_m (  -\tau )   )  \nonumber\\
   & =  (1 - \varepsilon )   e^{ \beta \Delta F }.
\end{eqnarray}
The final equality follows from
$F( \gamma )  =  - T \log Z$ and from the definition of $\varepsilon$.
\end{proof}
\noindent An analogous argument yields Eq.~\eref{eq:wEpsCalcn} in Lemma~\ref{lemma:WEpsCalcnCrooks}.

\section{Details of the work-extraction game}
\label{section:GameAppendix}
\subsection{Description of the game}
Let us briefly review the bound presented by Egloff et al.~\cite{EgloffDRV12}.
Consider the most efficient transformation $(\rho, H_\rho) \mapsto (\sigma, H_\sigma)$ that has a probability $1 - \delta$ of failing. 
That is, one sacrifices the certainty that the transformation will succeed, in hopes of extracting more work than can be gained from a certain-to-succeed transformation. All work that the system can output is collected; none is wasted. 
The transformation has a probability $1 - \delta$ of outputting at least the work 
\begin{equation}   \label{eq:EgloffDRV12Thm2}
   {w}^{\delta}_{\rm best} ( \rho, H_\rho   \mapsto   \sigma, H_\sigma )
   =  T  \log  \left(   M   \left(   
       \frac{  G^T(\rho) }{  1 - \delta  }  \:  ||  \:  G^T( \sigma )  
       \right)   \right).
\end{equation}
We will briefly overview the geometric definitions of \emph{Gibbs-rescaling} ($G^T$) and the \emph{relative mixedness} ($M$).

Let $\rho$ have the spectral decomposition
$\sum_{i = 1}^{ d_\rho } r_i  \ketbra{ E_i} {E_i}$ such that 
\begin{equation}
   r_1 e^{ \beta E_1 }   \geq   r_2 e^{ \beta E_2 } \geq \ldots   \geq   r_{ d_\rho }  e^{ \beta E_{ d_\rho } }.
\end{equation} 
Consider the histogram that represents the $r_i$. Gibbs-rescaling $\rho$ resizes each box in the histogram. The width of box $i$ changes from unity to $e^{ -\beta E_i }$, and the box's height increases by a factor of $e^{ \beta E_i }$. Denote by $h^T_\rho(u)$ the height of the point, on the rescaled histogram, whose $x$-coordinate is $u   \in [0, Z(H_\rho)]$. Integrating $h^T_\rho(u)$, we define the \emph{Gibbs-rescaled Lorenz curve} as the set of points
\begin{equation}
   \{ \left(   u,     L^T_\rho(u)      \right) \: | \:
       u   \in [0, Z(H_\rho)]    \},
   \quad {\rm wherein} \quad
   L^T_\rho(u)   \equiv   \int_{0}^u   h^T_\rho(u) du.
\end{equation}
The (unscaled) Lorenz curve $L_\rho$ is equivalent to $L_\rho^0$.
Upon Gibbs-rescaling $\rho$ and $\sigma$, we can compare the states' resourcefulness even though different Hamiltonians govern the states.

To incorporate the failure probability into the curve, we stretch $L^T_{ \rho }$ upward by a factor of
$1 / (1 - \delta)$. The resulting curve, $L^{T, \delta}_\rho$, encodes more reliable resourcefulness than $(\rho, H_\rho)$ possesses, because extractable work trades off with the failure probability $\delta$. Consider plotting $L^{T, \delta}_\rho$ on the same graph as $L^T_{\sigma}$. The curves are concave, bowing outward from the $x$-axis or stretching straight from $(0, 0)$ to $y = 1$. Consider compressing $L^T_\sigma$ leftward. $M$ denotes the inverse of the greatest factor by which 
$L^T_\sigma$ can compress without popping above $L^{T, \delta}_\rho$:
\begin{eqnarray}
   L^T_\rho(u)   & \geq &
   \left[ M   \left(   
       \frac{  G^T(\rho) }{  1 - \delta  }  \:  ||  \:  G^T( \sigma )  
       \right)    \right]^{-1}
    L^T_\sigma(u)
    \quad   \forall u \in [0,  \max ( Z( H_\rho ),  Z ( H_\sigma ) ]. \nonumber \\
&&
\end{eqnarray}
Illustrations appear in~\cite{EgloffDRV12}. While transforming $(\rho, H_\rho)$ into 
$(\sigma, H_\sigma)$, the player can extract no more work than $T \log M$: According to Eq.~(\ref{eq:EgloffDRV12Thm2}),
\begin{equation}   \label{eq:OriginalGameBound}
   w^\delta ( \rho, H_\rho   \mapsto   \sigma, H_\sigma )
   \leq  T  \log  \left(   M   \left(   
       \frac{  G^T(\rho) }{  1 - \delta  }  \:  ||  \:  G^T( \sigma )  
       \right)   \right).
\end{equation}

\subsection{Tightening and generalizing a one-shot bound with Crooks' Theorem}
\label{app:gametight}

Theorem~\ref{theorem:CrooksGamewEps2} strengthens the above Ineq.~\eref{eq:OriginalGameBound} (in the appropriate parameter regime), and Theorem~\ref{theorem:CrooksGameWeps}
generalizes this to work investment. 
These theorems are proved below:

%
%
%
%
\begin{theorem*} 
The work $\delta$-extractable from each Crooks-type reverse trial satisfies
\begin{equation}   \label{eq:CrooksGamew2}
   w^{\delta}   
   \leq \frac{1}{ \beta }   \left[
   \log M   \left(   
       \frac{  G^{T} ( \gamma_\tau ) }{  1 - \delta  }  \:  
       ||  \:  G^{T} (  \gamma_{-\tau} )  
       \right)
       -   H_\infty (P_{\fwd } ) 
   \right]
   \quad   \forall  \,  \delta \in [0, 1).
\end{equation}

\begin{proof}
During the reverse protocol, the state of a system $\mathcal{S}$ transforms as
\begin{equation}
   (    \gamma_{\tau},   H_\tau   )
   \mapsto
   ( \sigma,   H_{ -\tau } )
   \mapsto
   (   \gamma_{ - \tau },   H_{ -\tau } ),
\end{equation}
wherein $\sigma$ denotes some density operator that likely is not an equilibrium state.

The set of strategies for transforming from
$   (    \gamma_{\tau},   H_\tau   ) $ to
$   (   \gamma_{ - \tau },   H_{ -\tau } )$ 
contains a subset of all strategies that achieve this transformation via $( \sigma,   H_{ -\tau } )$.
As such, when optimizing to maximize the work production allowing an error rate of $\delta$,
\begin{eqnarray}
   w^{ \delta}_{\rm best} 
(    \gamma_{\tau},   H_\tau
      \mapsto
      \gamma_{ - \tau },   H_{ -\tau } )
   & \geq
   w^{ \delta}_{\rm best} 
      (    \gamma_{\tau},   H_\tau
      \mapsto   \sigma,   H_{ -\tau } ) \nonumber \\
   & \quad +
   w^0_{\rm best}
      (   \sigma,   H_{ -\tau }
      \mapsto   \gamma_{ - \tau },   H_{ -\tau } ).
\end{eqnarray}
The final term (which only involves thermalization) always succeeds, does not contribute either to the failure probability or the work cost, and hence can be eliminated.
An arbitrary, possibly suboptimal, strategy from    
$(    \gamma_{\tau},   H_\tau   )$
   to
  $ ( \sigma,   H_{ -\tau } )$
that generates work $w^\delta$ will be bounded by this value:
\begin{equation}
   w^{ \delta}   \leq
   w^{ \delta}_{\rm best} 
(    \gamma_{\tau},   H_\tau
      \mapsto
      \gamma_{ - \tau },   H_{ -\tau } ).
\end{equation}

Hence, by Eq.~(\ref{eq:EgloffDRV12Thm2}),
\begin{equation} \label{eq:ToGameBound1}
   w^{ \delta}
   \leq
   \frac{1}{\beta}   \log
   M   \left(   
       \frac{  G^T(  \gamma_{\tau}  ) }{  1 - \delta  }  \:  
       ||  \:  G^T(  \gamma_{ - \tau }  )  
       \right).
\end{equation}
Let us calculate $M$.
$L^T_{ \gamma_{ - \tau } }$ stretches straight from $(0, 0)$ to $(  Z_{ - \tau }, 1)$,  
whereas $L^T_{ \gamma_{\tau} } / (1 - \delta )$ stretches straight to 
$(  Z_{\tau} , \frac{1}{1 - \delta} )$. Compressing    
$L^T_{ \gamma_{\tau} }(u) / (1 - \delta) $ leftward by a factor of 
$ M^{-1}  =  \frac{ Z_{ \tau } (1 - \delta) }{   Z_{- \tau} }$ keeps the latter curve from dipping below $L^T_{ \gamma_{ - \tau } }(u)$. 
Hence
\begin{eqnarray}   
   \frac{1}{\beta} \log   M   \left(   
       \frac{  G^T(  \gamma_{\tau}  ) }{  1 - \delta  }  \:  
       ||  \:  G^T(  \gamma_{ - \tau }  )  
       \right)
   & =
   \frac{1}{\beta} \left[    \log \left(  
            \frac{ Z_{-\tau}  }{  Z_\tau }
            \right)  
          -  \log (1 - \delta) \right]   \nonumber \\
   & =   \label{eq:GameBoundw}
           \Delta F  -  \frac{1}{\beta}  \log (1 - \delta),
\end{eqnarray}
wherein $\Delta F   \equiv   F( \gamma_\tau)   -   F( \gamma_{-\tau} )$.
We substitute from Eq.~(\ref{eq:GameBoundw}) into the inequality~(\ref{eq:wEpsBound}) derived from Crooks' Theorem in Theorem~\ref{theorem:wDeltaBound}.
\end{proof}
\end{theorem*}
\noindent Crooks' Theorem introduces an $H_\infty$ into the bound, derived from~\cite{EgloffDRV12}, on extractable work. If $P_\fwd$ satisfies Ineq.~\eref{eq:GoodBound}, this $H_\infty$ strengthens the bound. 
A work cost can similarly be derived from~\cite{EgloffDRV12}, then enhanced with Crooks' Theorem.

\section{Details of thermodynamic resource theories}
\label{section:ResourceTheory}

\subsection{Description of thermodynamic resource theories}
Each thermodynamic resource theory models energy-preserving transformations performed with a heat bath characterized by an inverse temperature $\beta$.
To specify a state, one specifies a density operator and a Hamiltonian: $(\rho, H)$. Sums of Hamiltonians will be denoted by
$H_1  +  H_2  \equiv  H_1  \otimes  \id    +    \id \otimes H_2$. 

\emph{Thermal operations} can be performed for free. 
Each consists of three steps:
(1) A Gibbs state relative to $\beta$ and to any Hamiltonian $H_\gamma$ can be drawn from the bath:
\begin{equation}
   ( \gamma,   H_\gamma),
   \quad {\rm wherein} \quad
   \gamma  \equiv  \frac{  e^{ - \beta H_\gamma } }{ Z }.
\end{equation}
[Below, the Gibbs state relative to $\beta$ and to $H_\gamma$ will be denoted also by 
$\gamma( H_\gamma )$.]
Any unitary $U$ that conserves the total energy can be implemented, and any subsystem $A$ associated with its own Hamiltonian $H_A$ can be discarded. 
Each thermal operation on $(\rho, H)$ has the form
\begin{equation}
   (\rho, H)  \mapsto  
   \Big( {\rm Tr}_A   (    U    [ \rho   \otimes   \gamma ]    U^\dag ),   
 H + H_\gamma  - H_A
 \Big),
\end{equation}
wherein $[U, H + H_\gamma] = 0$.

\subsection{Applicability of Crooks' Theorem}
\label{app:BattAndClock}
Some resource-theory operations obey detailed balance and Markovianity.
We can use resource theories to model processes governed by Crooks' Theorem if we define a battery and a clock. 
Our model for the battery appears in~\cite{YungerHalpernR14} and resembles the model in~\cite{SkrzypczykSP14}. 
The semiclassical battery $B$ has closely spaced energy levels and occupies an energy eigenstate: 
\begin{equation}
   B_i   \equiv   ( \ketbra{ B_i }{ B_i },  H_B ),
   \quad {\rm wherein} \quad
   H_B  \equiv \sum_i  E_{B_i} \ketbra{ {B_i} }{ {B_i} }.
\end{equation}
If $E_{B_i}$ is large, a work-costing (forward) process can transfer work from the battery to $\mathcal{S}$. 
If $E_{B_i}$ is small, a work-extraction (reverse) process can transfer work from $\mathcal{S}$ to the battery.

We model the evolution of $H$ with a clock $C$ that occupies a pure state $\ket{C_j}$~\cite{BrandaoHORS13,HorodeckiO13}.
The changing of $\ket{C_j}$, like the movement of a clock hand, models the passing of instants. 
In processes governed by Crooks' Theorem, $H  =  H( \lambda_t )$. We discretize $t$ such that the system's Hamiltonian is $H( \lambda_{t_j} )$ when the clock occupies the state $\ket{C_j}$. In the notation introduced earlier, $t_1 = -\tau$, and $t_n = \tau$. The composite-system Hamiltonian 
\begin{equation}
   H_{\rm tot}   \equiv   
   \sum_{i = 1}^n   H( \lambda_{t_i} )    \otimes   \ketbra{ C_i }{ C_i }    
      \otimes \id
   +   \id   \otimes   \id   \otimes   \sum_j  E_{B_j} \ketbra{ B_j }{ B_j }
\end{equation} 
remains constant.

Having defined the battery and clock, we define the work extractable from, and the work cost of, a protocol. 
Let $E_{B_0} = 0$. The most work extractable from the reverse protocol equals the greatest $E_{B_m}$ for which some sequence of thermal operations evolves the state of $\mathcal{S}CB$ as
\begin{eqnarray}
   \gamma( H_{ - \tau } )   \otimes   \ketbra{ 1 }{ 1 }   \otimes   \ketbra{ 0 }{ 0 }
&   \mapsto
   \rho (t_2)   \otimes   \ketbra{ 2 }{ 2 }   \otimes   \ketbra{ E_{B_2} }{ E_{B_2} }
   \mapsto \ldots 
\nonumber \\ 
& \mapsto
 \gamma( H_\tau )   \otimes   \ketbra{ m }{ m }   \otimes   \ketbra{ E_{B_m} }{ E_{B_m} },
\end{eqnarray}
wherein $\rho (t_i)$ represents the state occupied by $\mathcal{S}$ at time $t_i$. 
The forward protocol's minimum work cost equals the least $E_{B_n}$ for which a sequence of thermal operations implements
\begin{equation}
   \gamma( H_\tau )   \otimes   \ketbra{ n }{ n }   \otimes   \ketbra{ E_{B_n} }{ E_{B_n} }
   \mapsto \ldots \mapsto
   \gamma( H_{ - \tau } )   \otimes   \ketbra{ 1 }{ 1 }   \otimes   \ketbra{ 0 }{ 0 }.
\end{equation}

Results in~\cite{HorodeckiO13} can be applied if all the states commute with their Hamiltonians. Horodecki and Oppenheim have calculated the maximum work yield, or minimum work cost, of any semiclassical transformation 
$(\rho, H_\rho)  \mapsto   (\sigma, H_\sigma)$ by thermal operations. They have calculated also faulty transformations' work yields and work costs. A faulty transformation generates a state 
$(\sigma', H_\sigma)$ that differs from the desired state.
The discrepancy is quantified by the trace distance between the density operators:
\begin{equation}
   \frac{1}{2} || \sigma - \sigma' ||_1  \leq  \epsilon \in [0, 1].
\end{equation} 
According to~\cite{HorodeckiO13}, this $\epsilon$ can be interpreted as the probability that the process fails to accomplish its mission, similarly to $\varepsilon$ and $\delta$. Generalizations to nonclassical states appear in~\cite{LostaglioJR15,LostaglioKJR15}.

\section{Details of numerical simulation}
\label{app:simulation}

Our simulation of Landauer bit reset and Szilard work extraction resembles the scenario presented in~\cite{BrowneGDV14}, that models a two-level quasiclassical system $\mathcal{S}$. 
At each time $t$, the energy $\mathcal{E}(t)$ of $\mathcal{S}$ equals $E_0$ or $E_1(t)$. 
The state of $\mathcal{S}$ is represented by a vector $\vec{s}(t) = ( p(t), 1 - p(t) )$, wherein $p(t)$ equals the probability that $\mathcal{E}(t)  =  E_0$.

If observers have different amounts of information about $\mathcal{E}(t)$, they ascribe different values to $p(t)$. 
Suppose an agent draws $\mathcal{S}$ from a temperature-$(1/ \beta)$ heat bath. 
According to this \emph{ignorant agent},
\begin{equation}
   \vec{s}(t)   =   \left(
   \frac{   e^{ - \beta E_0   }   }{   Z(t)   },
    \frac{   e^{ - \beta E_1(t)   }   }{   Z(t)   }   \right).
\end{equation}
According to an \emph{omniscient observer}, $\vec{s}(t) = (1, 0)$ or $(0, 1)$. The code is written from the perspective of an omniscient observer. On average, the code's predictions coincide with the predictions that code written by an ignorant agent would make.

While $t \in (-\infty, -\tau)$ during the forward (erasure) protocol, $E_1(t)  =  E_0  =  0$, and $\mathcal{S}$ is thermally equilibrated. According to the ignorant agent, $\vec{s}( t ) = ( \frac{1}{2}, \frac{1}{2} )$. Beginning at $t = - \tau$, the agent raises $E_1$ by the infinitesimal amount $dE$ while preserving $\vec{s}(t)$. Then, the agent couples $\mathcal{S}$ to the bath for some time interval. The raising and coupling are repeated until $t = \tau$ and $E_1(\tau) = E_{\rm max}$. 

The agent's actions are simulated as follows: Our code has a probability $\frac{1}{2}$ of representing the initial state $\vec{s}(-\tau)$ with $(1, 0)$ and a probability $\frac{1}{2}$ of representing $\vec{s}(-\tau)$ with $(0, 1)$. Consider one thermal interaction that occurs at some $t  \in  (- \tau,  \tau )$. If $\vec{s}( t ) = (1, 0)$ before the thermal interaction, the agent invests no work to raise $E_1$. If $\vec{s}( t) = (0, 1)$, the agent invests work $dE$.

A probabilistic swap models each interaction with the heat bath~\cite{BrowneGDV14}. $\vec{s}(t)$ has a probability $P_{\rm swap}$ of being exchanged with a pure state sampled from a Gibbs distribution. That is, $\vec{s}(t)$ has a probability 
$P_{\rm swap} e^{- \beta E_0 }  / { Z(t) } $ of being interchanged with $(1, 0)$, a probability 
$P_{\rm swap} e^{- \beta E_1(t) } / { Z(t) } $ of being interchanged with $(0, 1)$, and a probability
$1 - P_{\rm swap}$ of remaining unchanged: $\vec{s} ( t + dt ) = \vec{s} (t)$. The longer $\mathcal{S}$ couples to the reservoir, the greater the $P_{\rm swap}$.
The ignorant agent represents this thermalization with
$\vec{s}(t + dt)   =   M( t; P_{\rm swap} ) \vec{s}(t)$, wherein $M( t; P_{\rm swap} )$ is a thermalizing matrix that obeys detailed balance (see the proof of Lemma~\ref{lem:db_th_exist}).
Because $\vec{s}( t + dt )$ depends on no earlier state except $\vec{s}(t)$, the evolution is Markovian.

Ideally, the agent increases $E_1(t)$ and thermalizes $\mathcal{S}$ repeatedly until $t = \tau$, 
$E_1 (\tau) = \infty$, and $\vec{s}(\tau) = (1, 0)$ according to both observers. The simulated $E_1(t)$ peaks at some large $E_{\rm max}$, and the final state has a high probability of being $(1, 0)$~\cite{BrowneGDV14}. During stage two of erasure, $\mathcal{S}$ is thermally isolated, and $E_1$ decreases to zero. Because $\vec{s}$ has no weight on $E_1$, this stage costs no work.

Reversing erasure amounts to extracting work. Initially, $E_1 = E_0 = 0$, and $\vec{s} = (1, 0)$. As $\mathcal{S}$ remains thermally isolated, $E_1$ rises to infinity (approximated by $E_{\rm max}$) without consuming work. During stage two of work extraction, the agent repeatedly lowers $E_1(t)$ by $dE$ and thermalizes $\mathcal{S}$. Whenever the agent lowers $E_1(t)$ while $\vec{s}(t) = (0, 1)$ to the omniscient observer, $\mathcal{S}$ outputs work  $dE$. Once $t = -\tau$ such that $E_1(-\tau) = 0$, $\mathcal{S}$ thermalizes until the probability that $\vec{s}(t) = (1, 0)$ equals the probability that $\vec{s}(t) = (0, 1)$.

To produce Figure~\ref{fig:Simulation}, we simulated a bit-reset process where the energy gap between the two levels increases linearly from $0$ to $40 \, \kB T$ across $100,000$ equal time-steps.
After each step, the partial swap probability of thermalizing is $0.002$, and $\beta=10 \, \kB^{-1} K^{-1}$.
This process was repeated $10,000$ times, and the resulting work values binned into a histogram with $50$ divisions.
The maximum probability for the bit-reset work distribution was $P^{\rm max}=8.60$, satisfying $P^{\rm max} < \beta$.


\begin{thebibliography}{52}
\providecommand{\natexlab}[1]{#1}
\providecommand{\url}[1]{\texttt{#1}}
\expandafter\ifx\csname urlstyle\endcsname\relax
  \providecommand{\doi}[1]{doi: #1}\else
  \providecommand{\doi}{doi: \begingroup \urlstyle{rm}\Url}\fi

\bibitem[Jarzynski(1997)]{Jarzynski97}
C.~Jarzynski.
\newblock {Nonequilibrium Equality for Free Energy Differences}.
\newblock \emph{Physical Review Letters}, 78\penalty0 (14):\penalty0
  2690--2693, April 1997.
\newblock ISSN 0031-9007.
\newblock \doi{10.1103/PhysRevLett.78.2690}.
\newblock URL \url{http://link.aps.org/doi/10.1103/PhysRevLett.78.2690}.

\bibitem[Crooks(1999)]{Crooks99}
G.~E. Crooks.
\newblock {Entropy production fluctuation theorem and the nonequilibrium work
  relation for free energy differences}.
\newblock \emph{Physical Review E}, 60\penalty0 (3):\penalty0 2721--2726,
  September 1999.
\newblock ISSN 1063-651X.
\newblock \doi{10.1103/PhysRevE.60.2721}.
\newblock URL \url{http://link.aps.org/doi/10.1103/PhysRevE.60.2721}.

\bibitem[Kurchan(2000)]{Kurchan00}
J.~Kurchan.
\newblock {A Quantum Fluctuation Theorem}.
\newblock \emph{arXiv e-print}, July 2000.
\newblock URL \url{http://arxiv.org/abs/cond-mat/0007360}.

\bibitem[Tasaki(2000)]{Tasaki00}
H.~Tasaki.
\newblock {Jarzynski Relations for Quantum Systems and Some Applications}.
\newblock \emph{arXiv e-print}, September 2000.
\newblock URL \url{http://arxiv.org/abs/cond-mat/0009244}.

\bibitem[Engel and Nolte(2007)]{EngelN07}
A.~Engel and R.~Nolte.
\newblock {Jarzynski equation for a simple quantum system: Comparing two
  definitions of work}.
\newblock \emph{Europhysics Letters (EPL)}, 79\penalty0 (1):\penalty0 10003,
  July 2007.
\newblock ISSN 0295-5075.
\newblock \doi{10.1209/0295-5075/79/10003}.
\newblock URL \url{http://iopscience.iop.org/0295-5075/79/1/10003}.

\bibitem[Talkner et~al.(2007)Talkner, Lutz, and H\"{a}nggi]{TalknerLH07}
P.~Talkner, E.~Lutz, and P.~H\"{a}nggi.
\newblock {Fluctuation theorems: Work is not an observable}.
\newblock \emph{Physical Review E}, 75\penalty0 (5):\penalty0 050102, May 2007.
\newblock ISSN 1539-3755.
\newblock \doi{10.1103/PhysRevE.75.050102}.
\newblock URL \url{http://link.aps.org/doi/10.1103/PhysRevE.75.050102}.

\bibitem[Talkner and H\"{a}nggi(2007)]{TalknerH07}
P.~Talkner and P.~H\"{a}nggi.
\newblock {The Tasaki--Crooks quantum fluctuation theorem}.
\newblock \emph{Journal of Physics A: Mathematical and Theoretical},
  40\penalty0 (26):\penalty0 F569--F571, June 2007.
\newblock ISSN 1751-8113.
\newblock \doi{10.1088/1751-8113/40/26/F08}.
\newblock URL \url{http://iopscience.iop.org/1751-8121/40/26/F08/fulltext/}.

\bibitem[Quan and Dong(2008)]{QuanD08}
H.~T. Quan and H.~Dong.
\newblock {Quantum Crooks fluctuation theorem and quantum Jarzynski equality in
  the presence of a reservoir}.
\newblock \emph{arXiv e-print}, December 2008.
\newblock URL \url{http://arxiv.org/abs/0812.4955}.

\bibitem[Campisi et~al.(2009)Campisi, Talkner, and H\"{a}nggi]{CampisiTH09}
M.~Campisi, P.~Talkner, and P.~H\"{a}nggi.
\newblock {Fluctuation Theorem for Arbitrary Open Quantum Systems}.
\newblock \emph{Physical Review Letters}, 102\penalty0 (21):\penalty0 210401,
  May 2009.
\newblock ISSN 0031-9007.
\newblock \doi{10.1103/PhysRevLett.102.210401}.
\newblock URL \url{http://link.aps.org/doi/10.1103/PhysRevLett.102.210401}.

\bibitem[Talkner et~al.(2009)Talkner, Campisi, and H\"{a}nggi]{TalknerCH09}
P.~Talkner, M.~Campisi, and P.~H\"{a}nggi.
\newblock {Fluctuation theorems in driven open quantum systems}.
\newblock \emph{Journal of Statistical Mechanics: Theory and Experiment},
  2009\penalty0 (02):\penalty0 P02025, February 2009.
\newblock ISSN 1742-5468.
\newblock \doi{10.1088/1742-5468/2009/02/P02025}.
\newblock URL
  \url{http://iopscience.iop.org/1742-5468/2009/02/P02025/fulltext/}.

\bibitem[Campisi et~al.(2011)Campisi, H\"{a}nggi, and Talkner]{CampisiHT11}
M.~Campisi, P.~H\"{a}nggi, and P.~Talkner.
\newblock {Colloquium: Quantum fluctuation relations: Foundations and
  applications}.
\newblock \emph{Reviews of Modern Physics}, 83\penalty0 (3):\penalty0 771--791,
  July 2011.
\newblock ISSN 0034-6861.
\newblock \doi{10.1103/RevModPhys.83.771}.
\newblock URL \url{http://link.aps.org/doi/10.1103/RevModPhys.83.771}.

\bibitem[Esposito et~al.(2009)Esposito, Harbola, and S.]{EspositoHM09}
M.~Esposito, U.~Harbola, and Mukamel S.
\newblock {Nonequilibrium fluctuations, fluctuation theorems, and counting
  statistics in quantum systems}.
\newblock \emph{Reviews of Modern Physics}, 81\penalty0 (4):\penalty0
  1665--1702, December 2009.
\newblock ISSN 0034-6861.
\newblock \doi{10.1103/RevModPhys.81.1665}.
\newblock URL \url{http://link.aps.org/doi/10.1103/RevModPhys.81.1665}.

\bibitem[Hide and Vedral(2010)]{HideV10}
J.~Hide and V.~Vedral.
\newblock {Detecting entanglement with Jarzynski’s equality}.
\newblock \emph{Physical Review A}, 81\penalty0 (6):\penalty0 062303, June
  2010.
\newblock ISSN 1050-2947.
\newblock \doi{10.1103/PhysRevA.81.062303}.
\newblock URL \url{http://link.aps.org/doi/10.1103/PhysRevA.81.062303}.

\bibitem[Cohen and Imry(2012)]{CohenI12}
D.~Cohen and Y.~Imry.
\newblock {Straightforward quantum-mechanical derivation of the Crooks
  fluctuation theorem and the Jarzynski equality}.
\newblock \emph{Physical Review E}, 86\penalty0 (1):\penalty0 011111, July
  2012.
\newblock ISSN 1539-3755.
\newblock \doi{10.1103/PhysRevE.86.011111}.
\newblock URL \url{http://link.aps.org/doi/10.1103/PhysRevE.86.011111}.

\bibitem[Dorner et~al.(2013)Dorner, Clark, Heaney, Fazio, Goold, and
  Vedral]{DornerCHFGV13}
R.~Dorner, S.~R. Clark, L.~Heaney, R.~Fazio, J.~Goold, and V.~Vedral.
\newblock {Extracting Quantum Work Statistics and Fluctuation Theorems by
  Single-Qubit Interferometry}.
\newblock \emph{Physical Review Letters}, 110\penalty0 (23):\penalty0 230601,
  June 2013.
\newblock ISSN 0031-9007.
\newblock \doi{10.1103/PhysRevLett.110.230601}.
\newblock URL \url{http://link.aps.org/doi/10.1103/PhysRevLett.110.230601}.

\bibitem[Renner(2005)]{Renner05}
R.~Renner.
\newblock \emph{{Security of Quantum Key Distribution}}.
\newblock PhD thesis, ETH Z\"urich, December 2005.
\newblock URL \url{http://arxiv.org/abs/quant-ph/0512258}.

\bibitem[Dahlsten et~al.(2011)Dahlsten, Renner, Rieper, and
  Vedral]{DahlstenRRV11}
O.~C.~O. Dahlsten, R.~Renner, E.~Rieper, and V.~Vedral.
\newblock {Inadequacy of von Neumann entropy for characterizing extractable
  work}.
\newblock \emph{New Journal of Physics}, 13\penalty0 (5):\penalty0 053015, May
  2011.
\newblock ISSN 1367-2630.
\newblock \doi{10.1088/1367-2630/13/5/053015}.
\newblock URL \url{http://iopscience.iop.org/1367-2630/13/5/053015/fulltext/}.

\bibitem[Egloff et~al.(2012)Egloff, Dahlsten, Renner, and Vedral]{EgloffDRV12}
D.~Egloff, O.~C.~O. Dahlsten, R.~Renner, and V.~Vedral.
\newblock {Laws of thermodynamics beyond the von Neumann regime}.
\newblock \emph{arXiv e-print}, July 2012.
\newblock URL \url{http://arxiv.org/abs/1207.0434}.

\bibitem[\r{A}berg(2013)]{Aberg13}
J.~\r{A}berg.
\newblock {Truly work-like work extraction via a single-shot analysis.}
\newblock \emph{Nature communications}, 4:\penalty0 1925, January 2013.
\newblock ISSN 2041-1723.
\newblock \doi{10.1038/ncomms2712}.
\newblock URL
  \url{http://www.nature.com/ncomms/2013/130626/ncomms2712/full/ncomms2712.htm%
l}.

\bibitem[Horodecki and Oppenheim(2013)]{HorodeckiO13}
M.~Horodecki and J.~Oppenheim.
\newblock {Fundamental limitations for quantum and nanoscale thermodynamics.}
\newblock \emph{Nature Communications}, 4:\penalty0 2059, January 2013.
\newblock ISSN 2041-1723.
\newblock \doi{10.1038/ncomms3059}.
\newblock URL
  \url{http://www.nature.com/ncomms/2013/130626/ncomms3059/full/ncomms3059.htm%
l}.

\bibitem[Skrzypczyk et~al.(2014)Skrzypczyk, Short, and Popescu]{SkrzypczykSP14}
P.~Skrzypczyk, A.~J. Short, and S.~Popescu.
\newblock {Work extraction and thermodynamics for individual quantum systems.}
\newblock \emph{Nature communications}, 5:\penalty0 4185, January 2014.
\newblock ISSN 2041-1723.
\newblock \doi{10.1038/ncomms5185}.
\newblock URL
  \url{http://www.nature.com/ncomms/2014/140627/ncomms5185/full/ncomms5185.htm%
l}.

\bibitem[Brand\~{a}o et~al.(2015)Brand\~{a}o, Horodecki, Ng, Oppenheim, and
  Wehner]{BrandaoHNOW15}
F.~Brand\~{a}o, M.~Horodecki, N.~Ng, J.~Oppenheim, and S.~Wehner.
\newblock {The second laws of quantum thermodynamics}.
\newblock \emph{Proceedings of the National Academy of Sciences}, 112\penalty0
  (11):\penalty0 201411728, February 2015.
\newblock ISSN 0027-8424.
\newblock \doi{10.1073/pnas.1411728112}.
\newblock URL \url{http://www.pnas.org/content/112/11/3275}.

\bibitem[{Yunger Halpern} and {Renes}(2014)]{YungerHalpernR14}
N.~{Yunger Halpern} and J.~M. {Renes}.
\newblock {Beyond heat baths: Generalized resource theories for small-scale
  thermodynamics}.
\newblock \emph{ArXiv e-prints}, September 2014.
\newblock URL \url{http://arxiv.org/abs/1409.3998}.

\bibitem[Janzing et~al.(2000)Janzing, Wocjan, Zeier, Geiss, and
  Beth]{JanzingWZGB00}
D.~Janzing, P.~Wocjan, R.~Zeier, R.~Geiss, and T.~Beth.
\newblock {Thermodynamic Cost of Reliability and Low Temperatures: Tightening
  Landauer's Principle and the Second Law}.
\newblock \emph{International Journal of Theoretical Physics}, 39:\penalty0
  2717--2753, 2000.
\newblock \doi{10.1023/A:1026422630734}.
\newblock URL \url{http://dx.doi.org/10.1023/A%3A1026422630734}.

\bibitem[Horodecki et~al.(2009)Horodecki, Horodecki, Horodecki, and
  Horodecki]{HorodeckiHHH09}
Ryszard Horodecki, Pawe{\l} Horodecki, Micha{\l} Horodecki, and Karol
  Horodecki.
\newblock {Quantum entanglement}.
\newblock \emph{Rev. Mod. Phys.}, 81\penalty0 (2):\penalty0 865, 2009.
\newblock URL
  \url{http://journals.aps.org/rmp/abstract/10.1103/RevModPhys.81.865}.

\bibitem[Brand\~{a}o et~al.(2013)Brand\~{a}o, Horodecki, Oppenheim, Renes, and
  Spekkens]{BrandaoHORS13}
F.~G. S.~L. Brand\~{a}o, M.~Horodecki, J.~Oppenheim, J~Renes, and R.~W.
  Spekkens.
\newblock {Resource Theory of Quantum States Out of Thermal Equilibrium}.
\newblock \emph{Physical Review Letters}, 111\penalty0 (25):\penalty0 250404,
  December 2013.
\newblock ISSN 0031-9007.
\newblock \doi{10.1103/PhysRevLett.111.250404}.
\newblock URL \url{http://link.aps.org/doi/10.1103/PhysRevLett.111.250404}.

\bibitem[Szilard(1929)]{Szilard29}
L.~Szilard.
\newblock {\"{u}ber die Enfropieuerminderung in einem thermodynamischen System
  bei Eingrifen intelligenter Wesen}.
\newblock \emph{Zeitschrift f\"{u}r Physik}, 53:\penalty0 840--856, 1929.

\bibitem[Landauer(1961)]{Landauer61}
R~Landauer.
\newblock {Irreversibility and Heat Generation in the Computer Process}.
\newblock \emph{IBM Journal of Research and Development}, 5:\penalty0 183--191,
  1961.

\bibitem[Browne et~al.(2014)Browne, Garner, Dahlsten, and Vedral]{BrowneGDV14}
C.~Browne, A.~J.~P. Garner, O.~C.~O. Dahlsten, and V.~Vedral.
\newblock Guaranteed energy-efficient bit reset in finite time.
\newblock \emph{Phys. Rev. Lett.}, 113:\penalty0 100603, Sep 2014.
\newblock \doi{10.1103/PhysRevLett.113.100603}.
\newblock URL \url{http://link.aps.org/doi/10.1103/PhysRevLett.113.100603}.

\bibitem[Mossa et~al.(2009)Mossa, Manosas, Forns, Huguet, and
  Ritort]{MossaMFHR09}
A.~Mossa, M.~Manosas, N.~Forns, J.~M. Huguet, and F.~Ritort.
\newblock {Dynamic force spectroscopy of DNA hairpins: I. Force kinetics and
  free energy landscapes}.
\newblock \emph{Journal of Statistical Mechanics: Theory and Experiment},
  2009\penalty0 (02):\penalty0 P02060, February 2009.
\newblock ISSN 1742-5468.
\newblock \doi{10.1088/1742-5468/2009/02/P02060}.
\newblock URL
  \url{http://iopscience.iop.org/1742-5468/2009/02/P02060/fulltext/}.

\bibitem[{Manosas} et~al.(2009){Manosas}, {Mossa}, {Forns}, {Huguet}, and
  {Ritort}]{ManosasMFHR09}
M.~{Manosas}, A.~{Mossa}, N.~{Forns}, J.~M. {Huguet}, and F~{Ritort}.
\newblock {Dynamic force spectroscopy of DNA hairpins: II. Irreversibility and
  dissipation}.
\newblock \emph{Journal of Statistical Mechanics: Theory and Experiment},
  2009:\penalty0 P02061, February 2009.
\newblock \doi{10.1088/1742-5468/2009/02/P02061}.
\newblock URL \url{http://iopscience.iop.org/1742-5468/2009/02/P02061/tables}.

\bibitem[Alemany and Ritort(2010)]{AlemanyR10}
A.~Alemany and F.~Ritort.
\newblock {Fluctuation theorems in small systems: extending thermodynamics to
  the nanoscale}.
\newblock \emph{Europhysics News}, 41\penalty0 (2):\penalty0 27--30, April
  2010.
\newblock ISSN 0531-7479.
\newblock \doi{10.1051/epn/2010205}.
\newblock URL \url{http://dx.doi.org/10.1051/epn/2010205}.

\bibitem[Crooks(1998)]{Crooks98}
G.~E. Crooks.
\newblock {Nonequilibrium Measurements of Free Energy Differences for
  Microscopically Reversible Markovian Systems}.
\newblock \emph{Journal of Statistical Physics}, 90\penalty0 (5-6):\penalty0
  1481--1487, March 1998.
\newblock ISSN 1572-9613.
\newblock \doi{10.1023/A:1023208217925}.
\newblock URL \url{http://link.springer.com/article/10.1023/A:1023208217925}.

\bibitem[Jarzynski(2008)]{Jarzynski08}
C.~Jarzynski.
\newblock {Nonequilibrium work relations: foundations and applications}.
\newblock \emph{The European Physical Journal B}, 64\penalty0 (3-4):\penalty0
  331--340, July 2008.
\newblock ISSN 1434-6028.
\newblock \doi{10.1140/epjb/e2008-00254-2}.
\newblock URL
  \url{http://www.springerlink.com/index/10.1140/epjb/e2008-00254-2}.

\bibitem[Renner and Wolf(2004)]{RennerW04}
R.~Renner and S.~Wolf.
\newblock {Smooth R\'{e}nyi entropy and applications}.
\newblock In \emph{International Symposium on Information Theory, 2004. ISIT
  2004. Proceedings.}, pages 232--232. IEEE, 2004.
\newblock URL
  \url{http://ieeexplore.ieee.org/xpl/articleDetails.jsp?arnumber=1365269}.

\bibitem[Tomamichel(2012)]{Tomamichel12}
M.~Tomamichel.
\newblock \emph{{A Framework for Non-Asymptotic Quantum Information Theory}}.
\newblock PhD thesis, ETH Z\"urich, March 2012.
\newblock URL \url{http://arxiv.org/abs/1203.2142}.

\bibitem[Dupuis et~al.(2012)Dupuis, Kraemer, Faist, Renes, and
  Renner]{DupuisKFRR12}
F.~Dupuis, L.~Kraemer, P.~Faist, J.~M. Renes, and R.~Renner.
\newblock {Generalized Entropies}.
\newblock \emph{XVIIth International Congress on Mathematical Physics}, pages
  134--153, August 2012.
\newblock \doi{10.1142/9789814449243_0008}.
\newblock URL
  \url{http://www.worldscientific.com/doi/abs/10.1142/9789814449243_0008}.

\bibitem[{Lostaglio} et~al.(2015{\natexlab{a}}){Lostaglio}, {Jennings}, and
  {Rudolph}]{LostaglioJR15}
M.~{Lostaglio}, D.~{Jennings}, and T.~{Rudolph}.
\newblock {Description of quantum coherence in thermodynamic processes requires
  constraints beyond free energy}.
\newblock \emph{Nature Communications}, 6:\penalty0 6383, March
  2015{\natexlab{a}}.
\newblock \doi{10.1038/ncomms7383}.
\newblock URL
  \url{http://www.nature.com/ncomms/2015/150310/ncomms7383/full/ncomms7383.htm%
l}.

\bibitem[{Lostaglio} et~al.(2015{\natexlab{b}}){Lostaglio}, {Korzekwa},
  {Jennings}, and {Rudolph}]{LostaglioKJR15}
M.~{Lostaglio}, K.~{Korzekwa}, D.~{Jennings}, and T.~{Rudolph}.
\newblock {Quantum Coherence, Time-Translation Symmetry, and Thermodynamics}.
\newblock \emph{Physical Review X}, 5\penalty0 (2):\penalty0 021001, April
  2015{\natexlab{b}}.
\newblock \doi{10.1103/PhysRevX.5.021001}.
\newblock URL \url{http://dx.doi.org/10.1103/PhysRevX.5.021001}.

\bibitem[{Hanel} et~al.(2009){Hanel}, {Thurner}, and {Tsallis}]{HanelTT09}
R.~{Hanel}, S.~{Thurner}, and C.~{Tsallis}.
\newblock {On the robustness of q-expectation values and R{\'e}nyi entropy}.
\newblock \emph{EPL (Europhysics Letters)}, 85:\penalty0 20005, January 2009.
\newblock \doi{10.1209/0295-5075/85/20005}.
\newblock URL \url{http://iopscience.iop.org/0295-5075/85/2/20005/}.

\bibitem[Gour et~al.(2013)Gour, M\"{u}ller, Narasimhachar, Spekkens, and
  {Yunger Halpern}]{GourMNSYH13}
G.~Gour, M.~P. M\"{u}ller, V.~Narasimhachar, R.~W. Spekkens, and N.~{Yunger
  Halpern}.
\newblock {The resource theory of informational nonequilibrium in
  thermodynamics}.
\newblock \emph{arXiv e-print}, September 2013.
\newblock URL \url{http://arxiv.org/abs/1309.6586}.

\bibitem[{B\'{e}rut, A.} et~al.(2013){B\'{e}rut, A.}, {Petrosyan, A.}, and
  {Ciliberto, S.}]{BerutPC13}
{B\'{e}rut, A.}, {Petrosyan, A.}, and {Ciliberto, S.}
\newblock Detailed {J}arzynski equality applied to a logically irreversible
  procedure.
\newblock \emph{EPL}, 103\penalty0 (6):\penalty0 60002, 2013.
\newblock \doi{10.1209/0295-5075/103/60002}.
\newblock URL \url{http://dx.doi.org/10.1209/0295-5075/103/60002}.

\bibitem[{Jun} et~al.(2014){Jun}, {Gavrilov}, and {Bechhoefer}]{JunGB14}
Y.~{Jun}, M.~{Gavrilov}, and J.~{Bechhoefer}.
\newblock High-precision test of landauer's principle in a feedback trap.
\newblock \emph{Phys. Rev. Lett.}, 113:\penalty0 190601, Nov 2014.
\newblock \doi{10.1103/PhysRevLett.113.190601}.
\newblock URL \url{http://link.aps.org/doi/10.1103/PhysRevLett.113.190601}.

\bibitem[Alemany and Ritort(2013)]{AlemanyREmails}
A.~Alemany and F.~Ritort.
\newblock (Private communication), 2013.

\bibitem[Lacoste et~al.(2008)Lacoste, Lau, and Mallick]{LacosteLM08}
D.~Lacoste, A.~Lau, and K.~Mallick.
\newblock {Fluctuation theorem and large deviation function for a solvable
  model of a molecular motor}.
\newblock \emph{Physical Review E}, 78\penalty0 (1):\penalty0 011915, July
  2008.
\newblock ISSN 1539-3755.
\newblock \doi{10.1103/PhysRevE.78.011915}.
\newblock URL \url{http://link.aps.org/doi/10.1103/PhysRevE.78.011915}.

\bibitem[Cheng et~al.(2012)Cheng, Sreelatha, Hou, Efremov, Liu, van~der Maarel,
  and Wang]{ChengSHEL12}
Juan Cheng, Sarangapani Sreelatha, Ruizheng Hou, Artem Efremov, Ruchuan Liu,
  Johan R.~C. van~der Maarel, and Zhisong Wang.
\newblock Bipedal nanowalker by pure physical mechanisms.
\newblock \emph{Phys. Rev. Lett.}, 109:\penalty0 238104, Dec 2012.
\newblock \doi{10.1103/PhysRevLett.109.238104}.
\newblock URL \url{http://link.aps.org/doi/10.1103/PhysRevLett.109.238104}.

\bibitem[Serreli et~al.(2007)Serreli, Lee, Kay, and Leigh]{SerreliLKL07}
V.~Serreli, C.-F. Lee, E.~R. Kay, and D.~A. Leigh.
\newblock A molecular information ratchet.
\newblock \emph{Nature}, 445:\penalty0 523--527, 2007.
\newblock \doi{10.1038/nature05452}.
\newblock URL
  \url{http://www.nature.com/nature/journal/v445/n7127/abs/nature05452.html}.

\bibitem[MIT(2010)]{MITNano}
2010.
\newblock Baldo, Marc. 6.701 Introduction to Nanoelectronics, Spring 2010.
  (Massachusetts Institute of Technology: MIT OpenCourseWare),
  http://ocw.mit.edu License: Creative Commons BY-NC-SA.

\bibitem[Gomez-Marin et~al.(2008)Gomez-Marin, Parrondo, and {Van den
  Broeck}]{GomezMarinPvB08}
A.~Gomez-Marin, J.~M.~R. Parrondo, and C.~{Van den Broeck}.
\newblock {The “footprints” of irreversibility}.
\newblock \emph{EPL (Europhysics Letters)}, 82\penalty0 (5):\penalty0 50002,
  June 2008.
\newblock ISSN 0295-5075.
\newblock \doi{10.1209/0295-5075/82/50002}.
\newblock URL \url{http://stacks.iop.org/0295-5075/82/i=5/a=50002}.

\bibitem[{Dahlsten} et~al.(2015){Dahlsten}, {Choi}, {Braun}, {Garner}, {Yunger
  Halpern}, and {Vedral}]{DahlstenCBGYHV15}
O.~{Dahlsten}, M.-S. {Choi}, D.~{Braun}, A.~J.~P. {Garner}, N.~{Yunger
  Halpern}, and V.~{Vedral}.
\newblock {Equality for worst-case work at any protocol speed}.
\newblock \emph{ArXiv e-prints}, April 2015.
\newblock URL \url{http://arxiv.org/abs/1504.05152}.

\bibitem[{Salek} and {Wiesner}(2015)]{SalekW15}
S.~{Salek} and K.~{Wiesner}.
\newblock {Fluctuations in Single-Shot $\epsilon$-Deterministic Work
  Extraction}.
\newblock \emph{ArXiv e-prints}, April 2015.
\newblock URL \url{http://arxiv.org/abs/1504.05111}.

\bibitem[Horn and Johnson(1985)]{HornJ85}
R.~A. Horn and C.~R. Johnson.
\newblock \emph{{Matrix Analysis}}.
\newblock Cambridge University Press, Cambridge, 1985.
\newblock ISBN 0-521-38632-2.

\end{thebibliography}
\end{document}